%% file: main.tex
\documentclass[letter, 11pt, runningheads]{llncs}

\usepackage{setspace}
\usepackage[ruled,vlined]{algorithm2e}
\usepackage{amsfonts}
\usepackage{amsmath}
\usepackage[titletoc]{appendix}
\usepackage[english]{babel}
\usepackage{booktabs}
\usepackage{cancel}
\usepackage{caption}
\usepackage{comment}
\usepackage{diagbox}
\usepackage{dsfont}
\usepackage{enumerate}
\usepackage{eurosym}
\usepackage{float}
\usepackage[left = 1in, right = 1in]{geometry}
\usepackage{graphicx}
\usepackage[colorlinks=true,linkcolor=blue,citecolor=blue]{hyperref}
\usepackage{cleveref}
\usepackage{indentfirst}
\usepackage[utf8]{inputenc}\usepackage{csquotes} 
\usepackage{amssymb}
\usepackage{mathtext}
\usepackage{mathtools}
\usepackage{multicol}
\usepackage{multirow}
\usepackage{natbib}
\usepackage[short]{optidef}
\usepackage{pgfplots}
\usepackage{setspace}
\usepackage{subcaption}
\usepackage{supertabular}
\usepackage[flushleft]{threeparttable}
\usepackage{tikz,tkz-graph,tkz-berge}
\usetikzlibrary{decorations.pathreplacing}
\usetikzlibrary{arrows.meta}
\usepackage[textwidth=30mm]{todonotes}
\usepackage{xargs}
\usepackage{xcolor}
\usepackage[all]{xy}
\newcommand{\debug}[1]{{\color{black}#1}} 
\newcommand{\debugbis}[1]{{\color{black}#1}} 

\DeclareMathOperator{\Students}{\debug{\mathcal{S}}}
\DeclareMathOperator{\Colleges}{\debug{\mathcal{C}}}
\DeclareMathOperator{\Resources}{\debug{\mathcal{R}}}
\newcommandx{\CutoffsMatrix}{\debug{\pmb{\Theta}}}
\newcommandx{\Cutoffs}[2][1={}, 2={}]{\debug{\theta_{#1}^{\debugbis{#2}}}}
\DeclareMathOperator{\Contracts}{\debug{\mathcal{X}}}

\SetAlFnt{\small}
\SetAlCapFnt{\small}
\SetAlCapNameFnt{\small}
\SetAlCapHSkip{0pt}
\IncMargin{-\parindent}

\usepackage{pifont}

\let\oldnl\nl
\newcommand{\nonl}{\renewcommand{\nl}{\let\nl\oldnl}}

\setcitestyle{authoryear}

\begin{document}
\title{Two-Sided Matching with Resource-Regional Caps}


\author{\textbf{Felipe Garrido-Lucero}\inst{1,*} \and
\textbf{Denis Sokolov}\inst{2,*} \and \\
\textbf{Patrick Loiseau}\inst{2} \and \textbf{Simon Mauras}\inst{2}}
\institute{IRIT, Université Toulouse Capitole, Toulouse, France \and
Inria, Fairplay join team, Palaiseau, France
}
\maketitle              

\hspace{0.4cm}\small{* Equal Contribution} 

\hspace{0.4cm}\small{Preprint. Under review}

\begin{abstract}
We study two-sided many-to-one matching problems under a novel type of distributional constraints, resource-regional caps. In the context of college admissions, under resource-regional caps, an admitted student may be provided with a unit of some resource through a college, which belongs to a region possessing some amount of this resource. A student may be admitted to a college with at most one unit of any resource, i.e., all resources are close substitutes, e.g., dorms on the campus, dorms outside the campus, subsidies for renting a room, etc. The core feature of our model is that students are allowed to be admitted without any resource, which breaks heredity property of previously studied models with regions. 

It is well known that a stable matching may not exist under markets with regional constraints. Thus, we focus on three weakened versions of stability that restore existence under resource-regional caps: envy-freeness plus resource-efficiency, non-wastefulness, and novel direct-envy stability. For each version of stability we design corresponding matching mechanism(s). Finally, we compare stability performances of constructed mechanisms on an exhaustive collection of synthetic markets, and conclude that the most sophisticated direct-envy stable mechanism is the go-to mechanism for maximal stability of the resulting matching under resource-regional caps.

\end{abstract}

\input{Sections}

\bibliographystyle{ACM-Reference-Format}
\bibliography{sample}

\input{Appendix}

\end{document}

%% file: Sections.tex
\section{Introduction}

This paper develops a model of college admissions with resource-constrained regional caps (RRC), capturing the interplay between student assignment to programs and shared dormitory limitations.\footnote{See \Cref{app:motivation} for multiple real-life motivating examples.} The proposed framework generalizes classical matching theory to incorporate a joint allocation problem with regional or cross-program dormitory constraints, yielding both theoretical insights and policy-relevant applications.


The distinctive feature of our model is the possibility for a student to be admitted without being assigned any resource. This captures scenarios in which a student is content to attend a university without a dormitory placement -- e.g., if they live within commuting distance. For consistency, in such cases we say that the student receives an \textit{empty resource}, which is accessible by all colleges and has sufficient capacity to be given to any number of students. Such feature breaks the \textit{heredity} property \citep{kamada17,goto17}, broadly considered in the literature, which requires the feasibility of a matching to be monotone in the number of students matched.\footnote{See \Cref{sec:aziz} for a formal discussion of the heredity property.}

The foundational model for our work is matching under regional caps (RC) \citep{kamada10,kamada12,kamada15}. RRC constraints become RC if, first, all non-empty resources' regions are disjoint, and, second, any student should receive a unit of some non-empty resource upon admission. Therefore, all hardness results under RC transfers to RRC, e.g., non-existence of a stable matching \citep{kamada17}, and NP-completeness of checking the existence of a stable matching \citep{aziz24}. Note that, under no non-empty resources, RRC becomes classical \citet{gale62}.

Our analysis relies on another closely related model by \citet{aziz24}, who propose a summer internship program (SIP) matching framework with aggregate constraints over many possibly overlapping regions for many types of divisible resource (money coming from different supervisors), and impossibility for a student to be admitted without funding. While a reformulation that treats each college-resource pair as a distinct object restores heredity for RRC and enables analogies with SIP notions such as fairness and stability, this transformation loses core structural features of RRC model. In particular, RRC admits the presence of an empty resource and allows for \textit{indirect envy}, which are conceptually and analytically absent in SIP. As a result, while SIP tools such as cutoff-minimization remain applicable, a full stability analysis under RRC must account for these more general forms of feasibility and envy.\footnote{See \Cref{sec:aziz} for a formal comparison of RRC with \citet{aziz24}.}


Under RRC, all matchings can lack stability, so we study three milder notions that restore existence: (i) \textit{non-wastefulness}, consisting of \textit{seat-efficiency} and \textit{resource-efficiency}, (ii) \textit{envy-freeness}, consisting of elimination of \textit{direct} and \textit{indirect envy}, combined with resource-efficiency, and (iii) \textit{direct-envy stability}, which bans \textit{direct} envy over seat-resource bundles and any waste that would create it. Informally, we say that a student directly envies another student, if the latter has everything that the former requires: either a seat at a desired college together with a desired non-empty resource, or just a desired seat (which implies that the former student wants to be admitted with the empty resource). Other types of envy we call \textit{indirect}.

We design six mechanisms: (i) Random and Controlled Serial Dictatorships (RSD and CSD) which are non-wasteful, where CSD cleverly chooses the next student to be matched to minimize added envy; (ii) Increasing Uniform Cutoffs (IUC), which yields envy-free and resource-efficient outcomes by keeping all cutoffs equal for each college; (iii) Increasing Random, Increasing Minimal, and Increasing Deep Cutoffs (IRC, IMC, and IDC), which are direct-envy stable. In all these mechanisms, each college sets one cutoff per each resource. IMC keeps cutoffs as equal as possible, curbing envy further; IRC is fully random; whereas IDC is the adaptation of \cite{aziz24}'s mechanism to our setting. Direct-envy stability implies the \citet{kamada15} \textit{weak stability}. However, we establish impossibility results as no mechanism can guarantee direct-envy stability together with either seat-efficiency, resource-efficiency, or envy-freeness.

Simulations across varying numbers of resource types, alignment on preferences, and market structures show: without preference alignment, IMC is the most stable mechanism on average; when alignment exists, CSD remains both stable and strategy-proof. Given that real markets rarely align perfectly, IMC emerges as the most dependable choice under RRC.

The article is structured as follows. \Cref{sec:model} introduces two-sided many-to-one matching problem under RRC. \Cref{sec:matching} presents three versions of stability: envy-freeness plus resource-efficiency, non-wastefulness, and direct-envy stability; and shows three impossibility results. \Cref{sec:aziz} formally discusses differences and similarities of our model with \citet{aziz24}. \Cref{sec:cutoffs_and_cutoff_envy_free} introduces cutoffs. \Cref{section:mechanisms} constructs the six matching mechanisms and proves their properties. \Cref{sec:empirics} describes simulation results on proposed mechanisms. \Cref{sec:conclusion} concludes. Missing proofs are included in \Cref{app:proofs}, while \Cref{app:implementation_details,app:numerical_results} contain additional simulation results. \Cref{app:motivation} presents multiple real-life motivating examples and \Cref{app:literature} -- a detailed literature review. 

\section{Model}\label{sec:model}

We consider a finite set of \textbf{students} $\Students$, a finite set of \textbf{colleges} $\Colleges$, and a finite set of types of divisible \textbf{resources} $\Resources$. We consider, in addition, an \textit{empty resource} 
$r_0$ and denote $\Resources_0 := \Resources \cup \{r_0\}$. A \textbf{contract} is a triple of a student, a college, and a resource (empty or not) from the set of all contracts, denoted by $\Contracts := \Students \times \Colleges \times \Resources_0$.

Each college $c$ has a positive number of vacant seats $q_c>0$ to be distributed among admitted students, i.e., the \textbf{quota} of this college. We denote $\mathbf{q}_{\Colleges} := (q_c)_{c\in\Colleges}$ to be the quota profile of colleges.

Each resource has its own \textbf{regional cap}. Formally, we associate each type of resource $r\in\Resources_0$ to a non-empty subset of colleges $C_r \subseteq \Colleges$, representing a region where this resource can be distributed. In particular, for any resource $r$, the corresponding region $C_r$ has a maximum amount $q_r > 0$ of units of resource $r$ that can be allocated to the students admitted to colleges within $C_r$. Remark that regions may overlap, i.e., a college may have access to resources of different types. For an empty resource, we set $C_{r_0}=\Colleges$ and $q_{r_0}=|{\Students}|$, i.e., there are enough units for every college and every student. We denote $\mathbf{q}_{\Resources} := (q_r)_{r \in \Resources}$ to be the quota profile of non-empty resources.

Each college $c \in \Colleges$ has a strict priority ranking $\succ_c$ over the set of all students $\Students$. Whenever $s\,\succ_c s'$ for some $s,s'\in \Students$, college $c$ strictly prefers student $s$ to student $s'$.\footnote{Sometimes colleges may also have resource-specific priorities over students, e.g., room-eligibility rankings in the French college admissions platform, knows as Parcoursup.} We denote $\succ_{\Colleges} := \{\succ_c\}_{c\in\Colleges}$ to be the priority profile of colleges. 

Each student $s \in \Students$ has strict preferences over $\{\Colleges\times\Resources_0\} \cup \{\emptyset\}$, where, (i) $(c,r)\,\succ_s \emptyset$ (respectively, $\emptyset\succ_s (c,r)$) means that it is \textbf{acceptable} (respectively, unacceptable) for the student $s$ to be admitted to the college $c$ with one unit of the resource $r$ and (ii) 
$(c,r)\,\succ_s (c',r')$ means that student $s$ strictly prefers to be admitted to the college $c$ with a unit of the resource $r$ to being admitted to the college $c'$ with a unit of the resource $r'$. We denote a preference profile of students by $\succ_{\Students}:= \{\succ_s\}_{s\in\Students}$.

Note that contracts including the empty resource $r_0$ are never part of any (aggregate) resource constraints, thus, whenever a student can be admitted to a college with a non-empty resource, for sure she can be admitted at the same place with the empty one.\footnote{We formally state this property in \Cref{sec:matching} after introducing the notion of feasible matching.} Therefore, without loss of generality, we assume that for any student $s\in\Students$, any college $c\in\Colleges$, and any non-empty resource $r\in\Resources$, whenever $(c,r)$ is acceptable for $s$, it holds $(c,r)\succ_s (c,r_0)$. In other words, whenever a resource is acceptable for a student at a given college, the student prefers to get the resource to getting the empty one. Dropping this assumption does not affect the validity of our results.


Finally, we define a \textbf{matching market with resource-regional caps} (RRC) instance as any tuple $\langle \Students,\Colleges,\Resources_0,$ $\{C_r\}_{r\in \Resources}, \mathbf{q}_{\Colleges}, \mathbf{q}_{\Resources}, \succ_{\Colleges}, \succ_{\Students} \rangle.$

\section{Matching and Stability}\label{sec:matching}

A \textbf{matching} $\mu$ is any subset of the set of contracts $\Contracts$, such that \textbf{each student appears in at most one contract}. Given a matching $\mu$, we denote by $\mu_s$, $\mu_c$, and $\mu_r$, respectively, all triples within $\mu$ containing student $s$, college $c$, and resource $r$. Whenever a student $s$ is unmatched in $\mu$, we denote $\mu_s = \emptyset$. For a contract $\mu_s = (s,c,r) \in \mu$, we may abuse of notation and simply write $\mu_s = (c,r)$. In particular, we write $(c',r') \succ_s \mu_s$ whenever student $s$ prefers $(c',r')$ to $(c,r)$.

\begin{definition}
A matching $\mu$ is \textbf{feasible} if (i) each college $c$ has at most $q_c$ contracts and (ii) for any non-empty resource $r$, $\mu$ contains at most $q_r$ contracts containing it, all of them with colleges from $C_r$.
\end{definition}

The minimal property that we require for any matching is \textit{individual rationality}. 

\begin{definition}
A matching is \textbf{individually rational} if it does not contain unacceptable contracts. 
\end{definition}




The rest of this section is devoted to identify several notions of stability depending on the type of blocking contracts considered. \Cref{sec:stability} introduces \textit{stability} through \textit{envy-freeness} and \textit{non-wastefulness}. \Cref{sec:direct_envy_stability} introduces \textit{weak stability} and \textit{direct-envy stability}. Finally, \Cref{sec:resource_seat_eff} subdivides waste on \textit{resource-efficiency} and \textit{seat-efficiency}, and gives impossibility results.


\subsection{Stable Matching}\label{sec:stability}

Given a feasible and individually rational matching, we consider two classical sources of \textbf{instability}: (justified) envy and waste (wasted college seats and/or resources).

\begin{definition}\label{def:envy_freeness}
Let $\mu$ be an individually rational feasible matching. A contract $(s,c,r)\in\Contracts\backslash\mu$ is said to \textbf{envy-block} $\mu$ through a contract $(s',c,r')\in\mu$ if (i) $(c,r)\succ_s \mu_s$, (ii) $s\succ_c s'$, and (iii) the matching $(\mu\backslash\{\mu_s,(s',c,r')\})\cup \{(s,c,r)\}$ is feasible. A matching is called \textbf{envy-free} if it does not have envy-blocking contracts.
\end{definition}

In words, a contract $(s,c,r)$ envy-blocks a matching $\mu$ whenever $s$ prefers to be allocated at $c$ with the resource $r$, college $c$ prefers $s$ to one of its assigned students, and the matching remains feasible after making the replacement of contracts. Envy-free matchings are useful in many contexts (see \citet{kamada24}). Note that an empty matching is envy-free in our model. Thus, the existence of an envy-free matching is trivial.

\begin{definition}\label{def:non_wastefulness}
Let $\mu$ be a feasible and individually rational matching. A contract $(s,c,r) \in \Contracts \setminus \mu$ is said to \textbf{waste-block} $\mu$ if (i) $(c,r) \succ_s \mu_s$ and $(ii)$ $(\mu\backslash\{\mu_s\})\cup \{(s,c,r)\}$ is feasible. A matching is called \textbf{non-wasteful} if it does not have any waste-blocking contract.
\end{definition}

In words, a contract $(s,c,r)$ waste-blocks a matching $\mu$ whenever $s$ prefers to be allocated at some college $c$ with the resource $r$, and the matching stays feasible after changing their current contract to the more desired one. 

A contract is said to \textbf{block} a matching if it is either envy or waste-blocking. 

\begin{definition}\label{def:stability}
A feasible individually rational matching is called \textbf{stable} if it has no blocking contracts. 
\end{definition}

Unfortunately, stable matching may not exist, as exposed in \Cref{example:noStableMatching}.\footnote{The same example was firstly given by \citet{kamada17}.}

\begin{example}\label{example:noStableMatching}
Consider a market with two students, two colleges, and one (non-empty) resource with one region containing both colleges. Consider unit quotas $q_{c_1} = q_{c_2} = q_r = 1$ and the preferences given by $\succ_{s_1}: (c_1,r),(c_2,r),\emptyset$; $\succ_{s_2}: (c_2,r),(c_1,r),\emptyset$; $\succ_{c_1}: s_2,s_1$; and $\succ_{c_2}: s_1,s_2$.
There are five possible feasible individually rational matchings: $\mu_0=\{\}$, $\mu_1=\{(s_1,c_1,r)\}$, $\mu_2=\{(s_2,c_2,r)\}$, $\mu_3=\{(s_1,c_2,r)\}$, and $\mu_4=\{(s_2,c_1,r)\}$. However, note that $\mu_0$ is waste-blocked by $(s_2,c_1,r)$, $\mu_1$ is envy-blocked by $(s_2,c_1,r)$, $\mu_2$ is envy-blocked by $(s_1,c_2,r)$, $\mu_3$ is waste-blocked by $(s_1,c_1,r)$, and $\mu_4$ is waste-blocked by $(s_2,c_2,r)$. Thus, all five feasible matchings are unstable.
\end{example}

\subsection{Direct-Envy Stable Matchings}\label{sec:direct_envy_stability}

In order to guarantee the existence of some weakened version of stable matchings we relax the notion of stability by allowing for some kinds of envy-blocking and waste-blocking contracts to exist. The proposed version of weakened stability is inspired by the cutoff stability of \citet{aziz24}, and is built around a new notion of \textit{direct-envy}, which requires for the envious student to be able to obtain everything they need directly from another student.

\begin{definition}\label{def:direct_envy_blocking_contract}
Let $\mu$ be a feasible individually rational matching. A contract $(s,c,r)\in\Contracts\backslash\mu$ is said to \textbf{direct-envy-block} $\mu$ if it envy-blocks $\mu$ through a contract $(s',c,r') \in \mu$, such that, $r \in \{ r', r_0$\}. A matching is called \textbf{direct-envy-free} if it does not have direct-envy-blocking contracts.
\end{definition}

A contact is said to be \textbf{indirect-envy-blocking} if it is envy-blocking, but not direct-envy-blocking. With this in mind, we can adapt the notion of \textit{weak stability} from \citet{kamada15} to our setting.

\begin{definition}\label{def:weak_stability}
A feasible and individually rational matching is called \textbf{weakly stable} if it has no direct-envy-blocking contract and for any waste-blocking contract $(s,c,r)$ it holds (i) $r\neq r_0$, and (ii) in the region $C_r$ that contains $c$, all $q_r$ units of resource $r$ are distributed among admitted students.
\end{definition}

Weak stability is a well known stability notion in the literature of matching under distributional constraints \citep{aziz24,kamada17}. In line with \citet{aziz24}, we focus on studying a stronger stability notion in our setting, namely, \textit{direct-envy stability}. \Cref{prop:direct_envy_free_stability_implies_weak_stability} shows that direct-envy stability implies weak stability. Furthermore, \Cref{prop:IRC_IMC_IDC_des} implies that direct-envy stable matching always exists. To define direct-envy stability, we require the notion of dominance.

\begin{definition}\label{def:dominated_contract}
Let $\mu$ be a feasible and individually rational matching. A waste-blocking contract $(s,c,r)\in\Contracts\setminus\mu$ is said to be (direct-envy) \textbf{dominated} by another contract $(s',c,r')\in\Contracts\backslash\mu$ if (i) $(s',c,r')$ is neither waste-blocking nor direct-envy-blocking under $\mu$, and (ii) $(s',c,r')$ is direct-envy-blocking under $(\mu\backslash\{\mu_s\})\cup \{(s,c,r)\}$.
\end{definition}

\Cref{def:dominated_contract} states that a waste-blocking contract $x = (s,c,r) \in \Contracts\setminus\mu$ is dominated by another not direct-envy-blocking contract $x' \in \Contracts\setminus\mu$ whenever $x'$ becomes direct-envy blocking after replacing the contract $\mu_s$ within $\mu$ by $x$. From an implementation point of view, a central planner may prefer to leave some colleges with empty seats or undistributed resources, even if some students envy those places or resources, in order to avoid to trigger new direct-envy blocking contracts. A waste-blocking contact is said to be \textbf{undominated} if it is not dominated.




\begin{definition}\label{def:direct-envy-stability}
A feasible individually rational matching is called \textbf{direct-envy stable} if it has no direct-envy-blocking contracts and any waste-blocking contract is dominated. 
\end{definition}

\begin{proposition}\label{prop:direct_envy_free_stability_implies_weak_stability}
Any direct-envy stable matching is weakly stable.
\end{proposition}


\subsection{Seat and Resource Efficiencies - Impossibility Results}\label{sec:resource_seat_eff}

Traditionally, waste in matchings is due to empty seats in colleges. In our model, waste can also be generated by not allocating available resources. Indeed, for a matching $\mu$ that is  waste-blocked by a contract $(s,c,r)$, whether the student $s$ is already allocated at $c$ or somewhere else changes the source of wasting. We refer to a \textbf{seat-blocking} whenever $s$ is not already admitted to $c$ under $\mu$ or a \textbf{resource-blocking}, otherwise. A matching is called \textbf{seat-efficient} or \textbf{resource-efficient} if either there are no seat-blocking contracts or no resource-blocking contracts, respectively.




As we show later, a direct-envy stable matching always exists. A natural question arises after discovering this: \textit{can we always find a direct-envy stable matching that is either resource-efficient, seat-efficient, or envy-free?}. Surprisingly, the answer is no.

\begin{proposition}\label{prop:impossibility_resource} There exist markets without direct-envy stable and resource-efficient matchings.
\end{proposition}

\begin{proposition}\label{prop:impossibility_seat} There exist markets without direct-envy stable and seat-efficient matchings.
\end{proposition}

\begin{proposition}\label{prop:impossibility_envy} There exist markets without direct-envy stable and envy-free matchings.
\end{proposition}



\section{Heredity: Comparison with \citet{aziz24}}\label{sec:aziz}

A matching market model satisfies \textbf{heredity} whenever the feasibility of a matching is monotone in the number of agents matched to the objects. \citet{aziz24} proposed a \textit{summer internship problem} model (SIP) that satisfies heredity and involves overlapping regional constraints. \Cref{exmp:RRC_no_heredity} shows that RRC does not satisfy heredity. However, both models, matching markets with RRC constraints and matching markets à la \cite{aziz24} can be naturally mapped into the other one. In particular, given a matching markets with RRC constraints, it is enough to treat each college-resource pair as a project to obtain a SIP model.\footnote{If we treat each college and resource under RRC as a supervisor under SIP, then to admit a student each project should receive one unit of budget from each relevant supervisor, rather than from just one as under SIP. This technical difference between RRC and SIP can be avoided by assuming that the set of all feasible matchings was already calculated somehow and appears to be the same for both models.} With this in mind, the following equivalences can be established between the two models.
\vspace{-0.1cm}
\begin{itemize}
\item[$\bullet$] direct-envy-freeness under RRC and fairness under SIP;

\item[$\bullet$] non-wastefulness under RRC and strong non-wastefulness under SIP;

\item[$\bullet$] direct-envy stability under RRC and cutoff stability under SIP;

\item[$\bullet$] direct-envy-freeness plus non-wastefulness under RRC and strong stability under SIP.
\end{itemize}

Thus, despite the close relationship, these two models have a core difference that does not allow to treat RRC as a special case of SIP for stability analysis: unlike SIP, under RRC a student's envy may be \textit{indirect}, i.e., it may involve taking a seat from another student matched to some college-resource pair (some project) and using it for being matched to another college-resource pair (another project). In other words, under SIP there may not be exact analogs of an empty resource and an indirect envy, since these two features yield such possibilities for envy that are not present under strong stability of SIP.

As a result, in order to fully analyze stability properties of a matching under RRC we cannot exclusively focus on its analog that satisfies heredity. Still, we leverage the cutoff techniques from \citet{aziz24} to work with indirect stability.

\section{Cutoffs}\label{sec:cutoffs_and_cutoff_envy_free}

Inspired by \cite{aziz24}, we find direct-envy stable matchings by distributing resources among colleges in a \textit{balanced} way with respect to their preferences, through cutoffs.

\begin{definition}\label{def:cutoff_matrix}
A \textbf{cutoff profile} is a matrix $\CutoffsMatrix = (\Cutoffs[r][c])^{c\in\Colleges}_{r\in\Resources_0}$ where each $\Cutoffs[r][c]\in \{0,...,|{\Students}|\}$, such that, for any $c \in \Colleges$ and $r \in \Resources$, $\Cutoffs[r_0][c] = \Cutoffs[0][c] \geq \Cutoffs[r][c]$. A cutoff is \textbf{maximal} if it is equal to $|{\Students}|$.
\end{definition}

Cutoffs represent the \textit{top students} considered by a college as possible members. Formally, given a pair $(c,r) \in \Colleges\times \Resources_0$, such that $\Cutoffs[r][c] = k$, for $k \in \mathbb{N}$, we construct a matching only considering those contracts $(s,c,r) \in \Contracts$ for which $s$ is ranked on the $k$-top of $c$'s preferences. In particular, whenever $\Cutoffs[r][c] = 0$, no student can be allocated to $c$ with the resource $r$, while for $\Cutoffs[r][c] = |{\Students}|$, all contracts $\{(s,c,r), s \in \Students\}$ are possible to be included into the matching. Note that the condition $\Cutoffs[0][c] \geq \Cutoffs[r][c]$ ensures that whenever the contract $(s,c,r)$ can be included into the matching, $(s,c,r_0)$ can be included as well. We denote $\Contracts(\CutoffsMatrix)$ to the possible contracts to be considered for a matching according to the cutoffs profile $\CutoffsMatrix$.


\begin{definition}
Given a cutoff profile $\CutoffsMatrix$, we denote $\mu(\CutoffsMatrix)$ its \textbf{induced matching} obtained by including for each student $s$, their most preferred contract $(s,c,r)$ among those in $\Contracts(\CutoffsMatrix)$. 
\end{definition}

Note that $\mu(\CutoffsMatrix)$, a priory, \textbf{may not be feasible} as we may include to many contracts and not to respect the colleges or resources quotas. Moreover, even if the induced matching is feasible, it might be unstable. Finally, note that, given $\CutoffsMatrix$, computing $\mu(\CutoffsMatrix)$ can be done in polynomial time.

Given a cutoff profile $\CutoffsMatrix$, we denote $\CutoffsMatrix + \mathbf{1}_{r}^c$ the cutoff profile obtained by increasing $1$ unit to the entry $(c,r)$.

\begin{definition}
A cutoff profile $\CutoffsMatrix$ is said to be \textbf{optimal} if $\mu(\CutoffsMatrix)$ is feasible and for any non-maximal cutoff $\Cutoffs_{r}^c$, $\mu(\CutoffsMatrix + \mathbf{1}_{r}^c)$ is infeasible.
\end{definition}


\begin{proposition}\label{prop:optimal_cuts_des}
Any optimal cutoff profile induces a direct-envy stable matching. Conversely, any direct-envy stable matching can be induced by the (unique) optimal cutoff profile.
\end{proposition}




\section{Mechanisms}\label{section:mechanisms}



This section is devoted to present several matching mechanisms,
divided in two main families: cutoff mechanisms and serial dictatorships. Both families, starting from a (feasible) empty matching, update the current solution preserving its feasibility until no further update is possible. Therefore, the output matchings of the mechanisms are also always feasible.

\subsection{Cutoff Mechanisms}

Cutoff mechanisms, starting from the most restrictive cutoffs ($\CutoffsMatrix_0 \equiv 0$), create a sequence of increasing cutoff profiles until finding an optimal one. Once the final cutoff profile is computed, the induced matching is implemented. We present four different cutoffs mechanisms. 

\medskip

\noindent 1. \textbf{Increasing Random Cutoffs (IRC)}. The IRC mechanism, at each iteration $t$, uniformly picks a non-maximal cutoff and increases it in one unit if the resulting matching remains feasible. Note that, whenever the chosen cutoff $\Cutoffs[r][c]$ has the same value than $\Cutoffs[r_0][c]$, both of them are simultaneously increased to respect the condition on \Cref{def:cutoff_matrix}.
\medskip

\noindent 2. \textbf{Increasing Minimal Cutoffs (IMC)}. The IMC mechanism, at each iteration $t$, picks a permutation of the colleges $\sigma \in \Sigma(\Colleges)$ and, for each $c \in \sigma(\Colleges)$, it simultaneously increases the largest set of minimal-value cutoffs of $c$ that preserve the feasibility of the matching. By construction, the condition on \Cref{def:cutoff_matrix} is always verified.
\medskip

\noindent 3. \textbf{Increasing Deep Cutoffs (IDC)}. The IDC mechanism is the adaptation of \cite{aziz24} mechanism to our setting. At each iteration $t$, the mechanism picks a permutation of the cutoffs $\sigma \in \Sigma(\Colleges\times\Resources_0)$ and, for each $(c,r) \in \sigma(\Colleges\times\Resources_0)$, it increases $\Cutoffs[r][c]$ as much as possible without breaking the matching feasibility. Whenever $\Cutoffs[r][c] = \Cutoffs[r_0][c]$ and IDC wants to increase $\Cutoffs[r][c]$ by one unit, both cutoffs are simultaneously increased to respect the condition on \Cref{def:cutoff_matrix}.
\medskip

\noindent 4. \textbf{Increasing Uniform Cutoffs (IUC)}. The IUC mechanism, at each time $t$, picks $\sigma \in \Sigma(\Colleges)$ a permutation of the colleges, and for each $c \in \sigma(\Colleges)$, it increases all its cutoffs simultaneously, keeping all of them equal, whenever the resulting matching remains feasible. 
\medskip

Note that, given a feasible matching, the complexity of checking whether the matching remains feasible after either switching contracts for one student, or adding a new contract, is trivial and does not depend on the size of the market.\footnote{This occurs whenever the tentative matching changes under any of the mechanisms proposed in this paper.} Denote by $N=\max\{|{\Students}|,|{\Colleges}|,|{\Resources_0}|\}$. 

\begin{proposition}\label{prop:IRC_IMC_IDC_des}
IRC, IMC, and IDC always compute a direct-envy stable matching and run in $\mathcal{O}(N^3)$ time.
\end{proposition}

\begin{proposition}\label{prop:properties_IUC}
IUC always computes an  envy-free matching and runs in $\mathcal{O}(N^2)$ time.
\end{proposition}

\begin{definition}
A mechanism is called strategy-proof if no agent can obtain a strictly better contract by misreporting their preferences.    
\end{definition}

\begin{proposition}\label{prop:cutoff_mechs_not_strategy-proof}
IRC, IMC, IDC, and IUC are not strategy-proof and do not always find a stable matching whenever one exists.
\end{proposition}

\begin{proposition}\label{prop:cutoffs_mechanisms_are_stable_if_no_resources}
In a setting without non-empty resources IRC, IMC, IDC, and IUC are reduced to the same mechanism, which always computes a stable matching.
\end{proposition}

\subsection{Serial Dictatorships}

Serial Dictatorships mechanisms, starting from an empty matching, add one contract per student until all students have a contract or either all remaining contracts are unacceptable or no further contract can be included without breaking the feasibility of the matching. 

Given a matching $\mu$ and a student $s\in\Students$ that has no match under $\mu$, we define \textbf{the most preferred contract of student $s$ given the matching $\mu$}, denoted as $x^\mu_s$, as the the most preferred acceptable contract of $s$, such that $\mu \cup \{x^\mu_s\}$ is feasible. Note that $x^\mu_s$ may be empty. We present two mechanisms.
\medskip

\noindent 1. \textbf{Random Serial Dictatorship (RSD)}. Given a students permutation $\sigma \in \Sigma({\Students})$ as input, the RSD mechanism, at each iteration $t \in \{1,...,|{\Students}|\}$, computes the most preferred contract of the $t$-th student $s_t \in \sigma(\Students)$ given the current matching $\mu(t-1)$, and sets $\mu(t) = \mu(t-1) \cup x_{s_t}^{\mu(t-1)}$.
\medskip

\noindent 2. \textbf{Controlled Serial Dictatorship (CSD)}. The CSD mechanism, as the RSD mechanism, computes the most preferred contract of each student given the current matching and adds it. However, instead of prefixing the order in which students are included into the matching, at each iteration $t$, the CSD mechanism adds to the matching the highest ranked contract by some school among the unmatched students' most preferred feasible contracts at the moment, breaking ties uniformly. 



\begin{definition}
A matching $\mu$ is \textbf{Pareto-efficient} (for students) if there is no other feasible matching $\mu'$, such that, each student $s$ weakly prefers $\mu'_s$ to $\mu_s$, and at least one student strictly prefers $\mu'_s$ to $\mu_s$.
\end{definition}

\begin{proposition}\label{prop:properties_SD}
RSD and CSD always find non-wasteful and Pareto-efficient matchings, and run in $\mathcal{O}(N)$ time. RSD is strategy-proof, while CSD is not.
\end{proposition}

\section{Numerical Results}\label{sec:empirics}

We study the stability of the mechanisms introduced in \Cref{section:mechanisms} on an exhaustive collection of synthetic markets with agents having preferences with different levels of alignment, that is, agents agree on which contracts are more attractive. Note, however, that even under the full alignment regimes, agents can present different preferences rankings as, for each agent, we first sample the complete preference ordering and then randomly discard elements to simulate their unacceptable contracts. 

Additionally, we have considered balanced and unbalanced markets. In balanced markets, quotas add up the number of students. 
In unbalanced markets there could be shortcuts of students, colleges seats, or resources. The implementation details can be found in \Cref{app:implementation_details}. 

\Cref{tab:blocking_contracts_five_resources} presents the (average) number of each kind of blocking contracts of the matchings obtained with each mechanism. We remark the already known theoretical properties of the mechanisms: all cutoffs mechanisms produce direct-envy stable matchings, IUC produces envy-free matchings, and both serial dictatorship mechanisms produce no waste.

We remark two extra properties. First, although never completely stable, IMC presents the least number of blocking contracts. It is important to note that the presented perfect performance of IMC for colleges with full alignment on preferences does not always hold. However, and this is the second property to be remarked, CSD becomes fully stable under this regime. 

\begin{proposition}\label{prop:CSD_is_stable_under_colleges_alignment}
CSD is stable and strategy-proof whenever colleges' preferences are fully aligned.
\end{proposition}

\Cref{tab:blocking_contracts_two_resources_unbalanced_horizontal_market} presents the same metric for unbalanced markets with no alignment on preferences and two resources (one empty and one non-empty). Colleges up/down represents that the sum of seats over all colleges is double/half the total number of students. Resources up/down work similarly for the non-empty resource. As for balanced markets, IMC remains as the best mechanism. 

For complementary experiments on both balanced and unbalanced markets, please refer to \Cref{app:numerical_results}. To summarize both our theoretical and empirical findings, \Cref{fig:Guide_for_Policymakers} presents the \textbf{guide for policymakers} as the main takeaway of our article, that is aimed at maximizing the stability of the final matching under RRC, taking into account the policymaker's agenda.

\begin{figure}[H]
\centering
 \fbox{\parbox{.85\textwidth}{
    \includegraphics[scale = 0.85]{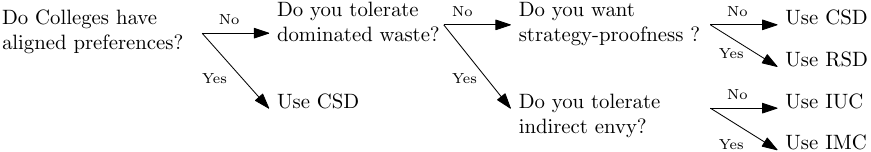}
    \caption{Guide for Policymakers}
    \label{fig:Guide_for_Policymakers}
}}
\end{figure}

\begin{table}[H]
\centering
\footnotesize{
\caption{Average numbers of blocking contracts over $100$ balanced markets with $100$ students, $10$ colleges, and $5$ resources $\Resources_0 = \{r_0,r_1,...,r_4\}$.}
\begin{tabular}{l|cccccc}
\toprule
Alignment & Mechanism & Resource & Seat   & Direct-Envy   & Indirect-Envy       & Total  \\
\hline
& IRC & 0.39±0.706  & 0.53±0.866     & 0.0±0.0       & 10.4±6.528  & 11.32±6.79\\
& IMC & 0.07±0.255  & 0.68±1.199     & 0.0±0.0       & 2.31±2.626  & \textbf{3.06±2.859}\\
No & IDC & 1.1±1.5     & 0.64±1.005     & 0.0±0.0       & 21.5±10.977 & 23.24±11.894\\
Alignment& IUC & 0.0±0.0     & 160.94±126.304 & 0.0±0.0       & 0.0±0.0     & 160.94±126.304\\
& RSD & 0.0±0.0     & 0.0±0.0        & 61.17±19.701  & 10.63±5.836 & 71.8±22.588\\
& CSD & 0.0±0.0     & 0.0±0.0        & 23.48±11.279  & 4.17±2.683  & 27.65±12.324\\
\toprule
& IRC & 1.73±2.315  & 1.03±1.758     & 0.0±0.0        & 13.06±10.469 & 15.82±11.646 \\
Students & IMC & 0.06±0.369  & 0.89±1.28      & 0.0±0.0        & 1.98±2.709   & \textbf{2.93±3.468 }\\
Semi & IDC & 4.8±5.255   & 1.61±2.825     & 0.0±0.0        & 36.06±28.787 & 42.47±32.336\\
Alignment & IUC & 0.0±0.0     & 354.27±321.113 & 0.0±0.0        & 0.0±0.0      & 354.27±321.113\\
& RSD & 0.0±0.0     & 0.0±0.0        & 213.12±110.211 & 23.66±22.408 & 236.78±100.5\\
& CSD & 0.0±0.0     & 0.0±0.0        & 47.28±23.946   & 7.15±8.476   & 54.43±26.635\\
\toprule
& IRC & 2.3±2.598   & 1.26±2.028     & 0.0±0.0       & 15.16±6.497   & 18.72±7.898\\
Student & IMC & 0.01±0.099  & 1.41±2.657     & 0.0±0.0       & 6.52±3.422    & \textbf{7.94±4.005}\\
Full & IDC & 10.33±9.483 & 2.47±3.64      & 0.0±0.0       & 43.86±19.592  & 56.66±26.388\\
Alignment& IUC & 0.0±0.0     & 1186.73±254.84 & 0.0±0.0       & 0.0±0.0       & 1186.73±254.84\\
& RSD & 0.0±0.0     & 0.0±0.0        & 488.88±71.458 & 153.24±27.468 & 642.12±83.045\\
& CSD & 0.0±0.0     & 0.0±0.0        & 23.99±9.436   & 7.69±4.762    & 31.68±13.009\\
\toprule
& IRC & 0.86±1.364  & 0.16±0.463    & 0.0±0.0       & 7.28±5.733   & 8.3±6.793 \\
College & IMC & 0.0±0.0     & 0.0±0.0       & 0.0±0.0       & 0.0±0.0      & \textbf{0.0±0.0}\\
Full & IDC & 1.74±2.23   & 0.61±0.882    & 0.0±0.0       & 20.36±13.481 & 22.71±15.196\\
Alignment & IUC & 0.0±0.0     & 154.0±144.681 & 0.0±0.0       & 0.0±0.0      & 154.0±144.681\\
& RSD & 0.0±0.0     & 0.0±0.0       & 62.73±23.231  & 9.61±5.181   & 72.34±24.671\\
& CSD & 0.0±0.0     & 0.0±0.0       & 0.0±0.0       & 0.0±0.0      & \textbf{0.0±0.0} \\
\toprule
& IRC & 4.54±4.579  & 0.44±1.219     & 0.0±0.0       & 33.15±11.332  & 38.13±14.432\\
College and &  IMC & 0.0±0.0     & 0.0±0.0        & 0.0±0.0       & 0.01±0.099    & \textbf{0.01±0.099}\\
Student & IDC & 13.95±9.907 & 1.31±2.208     & 0.0±0.0       & 84.9±27.937   & 100.16±36.121\\
Full &  IUC & 0.0±0.0     & 1006.75±210.17 & 0.0±0.0       & 0.0±0.0       & 1006.75±210.17\\
Alignment & RSD & 0.0±0.0     & 0.0±0.0        & 490.34±61.682 & 156.66±36.087 & 647.0±81.581\\
 & CSD & 0.0±0.0     & 0.0±0.0        & 0.0±0.0       & 0.0±0.0       & \textbf{0.0±0.0} \\
\hline
\end{tabular}
\label{tab:blocking_contracts_five_resources}
}
\end{table}

\begin{table}[ht]
\centering
\footnotesize{
\caption{Average numbers of blocking contracts over $100$ unbalanced markets without alignment with $100$ students, $10$ colleges, and $2$ resources $\Resources_0 = \{r_0,r_1\}$. }
\begin{tabular}{l|cccccc}
\toprule
Type of Unbalance & Mechanism & Resource & Seat   & Direct-Envy   & Indirect-Envy       & Total   \\
\hline
& IRC & 0.01±0.099  & 5.32±6.591   & 0.0±0.0       & 0.13±0.541      & \textbf{5.46±6.647}\\
& IMC & 0.0±0.0     & 5.4±6.713    & 0.0±0.0       & 0.08±0.504      & \textbf{5.48±6.798}\\
Colleges Up & IDC & 0.06±0.341  & 5.88±7.192   & 0.0±0.0       & 0.21±0.941      & 6.15±7.162\\
Resources Up & IUC & 0.0±0.0     & 143.4±172.39 & 0.0±0.0       & 0.0±0.0         & 143.4±172.39\\
& RSD & 0.0±0.0     & 0.0±0.0      & 39.09±42.618  & 0.0±0.0         & 39.09±42.618 \\
& CSD & 0.0±0.0     & 0.0±0.0      & 6.6±8.573     & 0.0±0.0         & 6.6±8.573  \\
\hline
& IRC & 1.81±1.809  & 0.0±0.0       & 0.0±0.0       & 2.82±2.304      & 4.63±3.882 \\
& IMC & 0.0±0.0     & 0.0±0.0       & 0.0±0.0       & 0.18±0.74       & \textbf{0.18±0.74} \\
Colleges Down & IDC & 7.29±4.401  & 0.0±0.0       & 0.0±0.0       & 10.41±5.972     & 17.7±9.829 \\
Resources Up & IUC & 0.0±0.0     & 48.12±115.202 & 0.0±0.0       & 0.0±0.0         & 48.12±115.202 \\
& RSD & 0.0±0.0     & 0.0±0.0       & 406.59±55.2   & 6.47±8.491      & 413.06±55.814 \\
& CSD & 0.0±0.0     & 0.0±0.0       & 78.87±21.933  & 1.29±2.151      & 80.16±22.562  \\
\hline
& IRC & 0.0±0.0     & 9.44±5.298    & 0.0±0.0       & 0.1±0.436       & \textbf{9.54±5.311} \\
& IMC & 0.0±0.0     & 9.36±5.289    & 0.0±0.0       & 0.18±0.477      & \textbf{9.54±5.294} \\
Colleges Up &  IDC & 0.0±0.0     & 9.36±89.194   & 0.0±0.0       & 0.18±0.0        & \textbf{9.54±89.194} \\
Resources Down & IUC & 0.0±0.0     & 356.75±89.194 & 0.0±0.0       & 0.0±0.0         & 356.75±89.194 \\
& RSD & 0.0±0.0     & 0.0±0.0       & 86.12±28.607  & 0.0±0.0         & 86.12±28.607 \\
& CSD & 0.0±0.0     & 0.0±0.0       & 11.93±8.171   & 0.0±0.0         & 11.93±8.171 \\
\hline
& IRC & 0.73±1.341  & 0.0±0.0       & 0.0±0.0       & 1.82±1.824      & 2.55±2.819 \\
& IMC & 0.08±0.578  & 0.0±0.0       & 0.0±0.0       & 0.84±1.111      & \textbf{0.92±1.454} \\
Colleges Down & IDC & 0.08±0.0    & 0.0±143.69    & 0.0±0.0       & 0.84±0.0        & \textbf{0.92±143.69} \\
Resources Down & IUC & 0.0±0.0     & 198.16±143.69 & 0.0±0.0       & 0.0±0.0         & 198.16±143.69 \\
& RSD & 0.0±0.0     & 0.0±0.0       & 378.88±68.604 & 2.86±6.547      & 381.74±69.471 \\
& CSD & 0.0±0.0     & 0.0±0.0       & 62.42±21.973  & 0.55±1.602      & 62.97±22.45 \\
\toprule
\end{tabular}
\label{tab:blocking_contracts_two_resources_unbalanced_horizontal_market}
}
\end{table}

\section{Conclusion}\label{sec:conclusion}

In this paper, we introduce a novel kind of distributional constraints, resource-regional caps (RRC). Since a stable matching may not exist in general, we propose three weaker versions of stability that restore existence: envy-freeness plus resource efficiency, non-wastefulness (seat- and resource-efficiency), and direct-envy stability. The latter is based on the idea of a direct envy of one student towards another, which arises if the second student, object of envy, has everything that the first one, subject of envy, requests. A direct-envy stable matching does not tolerate direct envy and any waste that does not create more direct envy upon being resolved. Then, we show that direct-envy stability is incompatible with either seat-/resource-efficiency, or envy-freeness.

For each weaker version of stability we design at least one corresponding mechanism: Increasing Uniform Cutoffs (IUC) is envy-free and resource-efficient, Random and Controlled Serial Dictatorships (RSD and CSD) are non-wasteful, and Increasing Random/Deep/Minimal Cutoffs (IRC, IDC, and IMC) are direct-envy stable. CSD and IMC are designed to be more stable (on average) versions of, correspondingly, RSD and IRC. 

In addition, we test stability performances of all six mechanisms under (i) different numbers of non-empty resources types, (ii) several types of preferences alignment: no alignment, partial alignment, and full alignment, and (iii) balanced and unbalanced markets. As a result, theoretical and numerical findings suggest that, under no alignment of colleges' preferences, IMC is the most stable on average, while otherwise, CSD is stable and strategy-proof. Thus, given that in the real world most markets have no perfect alignment of preferences, IMC is the go-to mechanism for maximizing stability of the resulting matching under RRC constraints.

Finally, we propose some potential avenues for the future research on RRC. Is it possible to strengthen the stability of existing matchings beyond the direct-envy stability?

\Cref{prop:impossibility_resource,prop:impossibility_seat,prop:impossibility_envy} show that the answer is negative with current types of blocking contracts. Thus, one may disaggregate sets of blocking contracts even further. For instance, there could be defined two types of indirect envy: the candidate with indirect envy either (i) already has a needed resource, or (ii) requests an empty resource from the college. 

Going back to real-life motivation, French college admissions not only operate with common dormitories but also use affirmative action \citep{sokolov25}, while there is no (at least known to us) research that combines these two features. Moreover, the French college admissions system includes dormitory-specific rankings of students. Incorporating resource-specific priority structures within each region into the RRC framework would be a promising direction for further research.

%% file: Appendix.tex
\appendix

\section{Missing Proofs}\label{app:proofs}




As discussed in \Cref{sec:aziz}, both models of matching under RRC and SIP can be naturally mapped into the other one. However, from RRC to SIP the modeling must be carefully done, as taking students as agents and colleges as objects, without considering resources, leads to a model without hereditary property, as the following example illustrates.

\begin{example}\label{exmp:RRC_no_heredity}
Consider a market with five students, one college, and one non-empty resource. The college has exactly three seats to distribute. The region containing the only college has exactly one unit of resource to distribute. Students $s_1$, $s_2$, and $s_3$ do not need a unit of a non-empty resource $r$ for admission, while students $s_4$ and $s_5$ can be admitted only with a unit of $r$.

Consider $X'=\{(s_1,c_1,r_0),(s_2,c_1,r_0),(s_3,c_1,r_0)\}$ and $X''=\{(s_4,c_1,r),(s_5,c_1,r)\}$. First, $|X''|<|X'|$, second, $X'$ is feasible, while $X''$ is not. Thus, RRC constraints do not satisfy heredity. \qed
\end{example}

We give next the missing proofs of the main body.
\medskip

\noindent\textit{Proof of \Cref{prop:direct_envy_free_stability_implies_weak_stability}}. Take a direct-envy stable matching $\mu$. By definition, there is no direct-envy-blocking contracts, and any waste-blocking contract is dominated. Let $(s,c,r) \in\Contracts\backslash\mu$ be a waste-blocking contract. We need to prove that $r\neq r_0$ and that in the region that contains college $c$, all units of resource $r$ are distributed among admitted students.

    Since $\mu$ is direct-envy stable, $(s,c,r)$ is dominated, i.e., there exists another contract $(s',c,r')\in\Contracts\backslash\mu$, such that (i) $(s',c,r')$ is not waste-blocking or direct-envy-blocking under $\mu$ and (ii) $(s',c,r')$ is direct-envy-blocking under $(\mu\backslash\{\mu_s\})\cup \{(s,c,r)\}$. In particular, since it is $(s,c,r)$ that makes $(s',c,r')$ direct-envy-blocking under $(\mu\backslash\{\mu_s\})\cup \{(s,c,r)\}$, then $r = r' \neq r_0$, as whenever $r' = r_0$, $(s',c,r')$ should be either direct-envy-blocking or waste-blocking under $\mu$.

    Suppose now that under $\mu$ in the region containing $c$, some units of resource $r$ are not distributed. Hence, $(s',c,r')$ should be either direct-envy-blocking or waste-blocking under $\mu$, since $r' = r$, which is a contradiction. We conclude that matching $\mu$ is weakly stable.
\qed\medskip

\noindent\textit{Proof of \Cref{prop:impossibility_seat}}. Consider the market in \Cref{example:noStableMatching}. It is not hard to see that both direct-envy stable matchings $\mu_1$ and $\mu_2$ have seat-blocking contracts.
\qed\medskip

\noindent\textit{Proof of \Cref{prop:impossibility_resource}}. Consider a market with three students, three colleges, and one (non-empty) resource $r$ with region containing all colleges. Quotas are $q_{c_1} = q_{c_2} = q_{c_3} = q_r = 1$. Preferences are:

    \begin{table}[H]
        \centering
        \begin{tabular}{l c c}
            $\succ_{s_1}: (c_1,r),(c_1,r_0),(c_3,r),(c_2,r_0),(c_3,r_0),\emptyset$ &  & $\succ_{c_1}: s_3,s_2,s_1$\\
            $\succ_{s_2}: (c_1,r),\emptyset$ &  & $\succ_{c_2}: s_1,s_2,s_3$\\
            $\succ_{s_3}: (c_2,r_0),(c_1,r_0),\emptyset$ &  & $\succ_{c_3}: s_3,s_1,s_2$
        \end{tabular}
    \end{table}

    Let us look for all direct-envy stable matchings. Take a direct-envy stable matching $\mu$.
    
    Suppose that $(s_2,c_1,r)\in \mu$. Thus, none of the following contracts is in $\mu$ due to feasibility: $(s_3,c_1,r_0),(s_1,c_1,r)$,$(s_1,c_1,r_0),(s_1,c_3,r)$. Also, if $(s_3,c_2,r_0)\in \mu$, then $s_1$ has a direct envy towards $s_3$, so $(s_3,c_2,r_0)\not\in \mu$. Therefore, the matching should be $\mu=\{(s_1,c_2,r_0),(s_2,c_1,r)\}$, which is not direct-envy stable: $s_3$ direct envies $s_1$ with $(s_3,c_1,r_0)$. Therefore, $(s_2,c_1,r)\not\in \mu$.

    So far, we know that $(s_2,c_1,r)\not\in \mu$, i.e., $s_2$ is unmatched. Thus, $(s_1,c_1,r)\not\in \mu$, since otherwise $s_2$ direct envies $s_1$. Suppose that $(s_3,c_1,r_0)\in \mu$. Thus, none of the following contracts is in $\mu$ due to feasibility: $(s_3,c_2,r_0),(s_1,c_1,r_0)$. Thus, the matching should be $\mu=\{(s_1,c_3,r),(s_3,c_1,r_0)\}$, which is not direct-envy stable: $(s_3,c_2,r_0)$ is undominated waste-blocking contract. Therefore, $(s_3,c_1,r_0)\not\in \mu$.

    So far, we know that none of the following contracts is in $\mu$: $(s_1,c_1,r),(s_2,c_1,r),(s_3,c_1,r_0)$. Suppose that $(s_1,c_1,r_0)\in \mu$. Thus, none of the following contracts is in $\mu$ due to feasibility: $(s_1,c_2,r_0),(s_1,c_3,r_0),(s_1,c_3,r)$. Thus, the matching should be $\mu=\{(s_1,c_1,r_0),(s_3,c_2,r_0)\}$, which is direct-envy stable, but not resource-efficient: $(s_1,c_1,r)$ is resource-blocking.

    So far, we know that none of the following contracts is in $\mu$: $(s_1,c_1,r),(s_1,c_1,r_0),(s_2,c_1,r),$ $(s_3,c_1,r_0)$. Suppose that $(s_3,c_2,r_0)\in \mu$. Thus, the matching should be $\mu=\{(s_1,c_3,r),(s_3,c_2,r_0)\}$, which is not direct-envy stable: $(s_1,c_1,r_0)$ is undominated waste-blocking contract. Therefore, $(s_3,c_2,r_0)\not\in \mu$.

    So far, we know that none of the following contracts is in $\mu$: $(s_1,c_1,r),(s_1,c_1,r_0),(s_2,c_1,r),$ $(s_3,c_1,r_0),(s_3,c_2,r_0)$. Thus, the matching should be $\mu=\{(s_1,c_3,r)\}$, which is not direct-envy stable: $(s_3,c_2,r_0)$ is undominated waste-blocking contract.
    
    Therefore, there is only one direct-envy stable matching $\mu=\{(s_1,c_1,r_0),(s_3,c_2,r_0)\}$, which is not resource-efficient.
\qed\medskip

\noindent\textit{Proof of \Cref{prop:impossibility_envy}}. Consider a market with three students, three colleges, and one (non-empty) resource $r$ with region containing all colleges. Quotas are $q_{c_1} = q_{c_2} = q_{c_3} = q_r = 1$. Preferences are:

    \begin{table}[H]
        \centering
        \begin{tabular}{l c c}
            $\succ_{s_1}: (c_1,r),(c_1,r_0),(c_3,r),(c_2,r),(c_3,r_0),(c_2,r_0),\emptyset$ &  & $\succ_{c_1}: s_2,s_3,s_1$\\
            $\succ_{s_2}: (c_3,r),(c_2,r),(c_1,r),(c_2,r_0),(c_3,r_0),(c_1,r_0),\emptyset$ &  & $\succ_{c_2}: s_2,s_3,s_1$\\
            $\succ_{s_3}: (c_1,r),(c_3,r),(c_1,r_0),(c_2,r),\emptyset$ &  & $\succ_{c_3}: s_3,s_2,s_1$
        \end{tabular}
    \end{table}

    Let us look for all direct-envy-stable matchings. Since there is only one unit of non-empty resource, there can be one of four cases.
    \begin{enumerate}
        \item Suppose no-one gets a non-empty resource $r$. Then, either $(s_3,r_0)$ is admitted to $c_1$, or $s_3$ is unmatched.
        \begin{enumerate}
            \item If $(s_3,r_0)$ is admitted to $c_1$, then $(s_2,r_0)$ is either admitted to $c_2$ or to $c_3$ (to avoid undominated waste). However, if $(s_2,r_0)$ is admitted to $c_i$ for any $i\in\{2,3\}$, then $(s_2,c_i,r)$ is an undominated waste-blocking contract. 
            \item If $s_3$ is unmatched, then $(s_2,r_0)$ is either admitted to $c_1$, or to $c_2$, or to $c_3$ (to avoid undominated waste). 
            \begin{enumerate}
                \item If $(s_2,r_0)$ is admitted to $c_i$ for any $i\in\{2,3\}$, then $(s_2,c_i,r)$ is an undominated waste-blocking contract.
                \item If $(s_2,r_0)$ is admitted to $c_3$, then $(s_1,r_0)$ is either admitted to $c_i$ for some $i\in\{1,2\}$ or unmatched.
                \begin{enumerate}
                    \item If $(s_1,r_0)$ is admitted to $c_1$, then $(s_3,c_1,r_0)$ is a direct-envy-blocking contract.
                    \item If $(s_1,r_0)$ is either admitted to $c_2$ or unmatched, then $(s_3,c_1,r_0)$ is an undominated waste-blocking contract.
                \end{enumerate}
            \end{enumerate}
        \end{enumerate}
        \item Suppose that $s_1$ gets a unit of $r$. Then each of $s_2$ and $s_3$ either has $r_0$ or is unmatched, and $(s_1,r)$ is matched to some $c_i$ with $i\in\{1,2,3\}$. However, if $(s_1,r)$ is admitted to $c_i$ for any $i\in\{1,2,3\}$, then $(s_2,c_i,r)$ is a direct-envy-blocking contract.

        \item Suppose that $s_2$ gets a unit of $r$. Then each of $s_1$ and $s_3$ either has $r_0$ or is unmatched, and $(s_2,r)$ is matched to some $c_i$ with $i\in\{1,2,3\}$.
        \begin{enumerate}
            \item If $(s_2,r)$ is admitted to $c_1$, then $s_3$ is unmatched, then $(s_1,r_0)$ is matched to $c_3$. So, $(s_2,c_2,r)$ is an undominated waste-blocking contract.

            \item If $(s_2,r)$ is admitted to $c_2$, then either $(s_3,r_0)$ is matched to $c_1$ or $s_3$ is unmatched.

            \begin{enumerate}
                \item If $(s_3,r_0)$ is matched to $c_1$, then $(s_1,r_0)$ is matched to $c_3$. So, $\{(s_1,c_3,r_0),(s_2,c_2,r),$ $(s_3,c_1,r_0)\}$ is \textbf{the first direct-envy stable matching}. There is an \textbf{indirect-envy-blocking contract} $(s_2,c_3,r)$.

                \item If $s_3$ is unmatched, then $(s_1,r_0)$ is matched to $c_3$. So, $(s_3,c_1,r_0)$ is an undominated waste-blocking contract.
            \end{enumerate}
        \end{enumerate}

        \item Finally, suppose that $s_3$ gets a unit of $r$. Then each of $s_1$ and $s_2$ either has $r_0$ or is unmatched, and $(s_3,r)$ is matched to some $c_i$ with $i\in\{1,2,3\}$.
        \begin{enumerate}
            \item If $(s_3,r)$ is admitted to $c_i$ for any $i\in\{1,2\}$, then $(s_2,c_i,r)$ is a direct-envy-blocking contract.

            \item If $(s_3,r)$ is admitted to $c_3$, then either $(s_1,r_0)$ is matched to some $c_i$ with $i\in\{1,2\}$, or $s_1$ is unmatched.

            \begin{enumerate}
                \item If $(s_1,r_0)$ is matched to $c_1$, then $(s_2,r_0)$ is matched to $c_2$. So, $\{(s_1,c_1,r_0),(s_2,c_2,r_0),$ $(s_3,c_3,r)\}$ is \textbf{the second direct-envy stable matching}. There is an \textbf{indirect-envy-blocking contract} $(s_3,c_1,r)$.

                \item If $(s_1,r_0)$ is matched to $c_2$, then $(s_2,r_0)$ is matched to $c_1$. So, $(s_2,c_2,r_0)$ is a direct-envy-blocking contract.

                \item If $s_1$ is unmatched, then $(s_2,r_0)$ is matched to $c_2$. So, $(s_1,c_1,r_0)$ is an undominated waste-blocking contract.
            \end{enumerate}
        \end{enumerate}
    \end{enumerate}
    Therefore, there are only two direct-envy stable matchings $\{(s_1,c_3,r_0),(s_2,c_2,r),$ $(s_3,c_1,r_0)\}$ and $\{(s_1,c_1,r_0),(s_2,c_2,r_0),$ $(s_3,c_3,r)\}$, none of which is envy-free.
    
    This concludes the proof.
\qed\medskip

\noindent\textit{Proof of \Cref{prop:optimal_cuts_des}}. By definition, the matching induced by an optimal profile of cutoffs is such that, for each college holds, that none of its cutoffs can be increased by one in order to get a feasible matching (making sure that the empty cutoff is always the highest). Take such matching $\mu$.
    
    If $(s,c,r)\in \mu$, then any $(s',c,r)$, such that $s'\succ_c s$ could have been included into $\mu$. Moreover, if $(s,c,r_0)\in \mu$, then any $(s',c,r)$, such that $r\in\Resources$ and $s'\succ_c s$ could have been included into $\mu$ Therefore, there are no direct-envy-blocking contracts under $\mu$.
    
    Now suppose that there is a seat-blocking contract $(s,c,r)$ under $\mu$. First, we show that $r\neq r_0$. Since $c$ has an open seat, if $r=r_0$, then we would get feasible matching my setting $\Cutoffs[0][c]$ equal to the rank of $s'$ in ranking of $c$, where $s'$ is the highest ranked student with a seat-blocking contract $(s',c,r_0)$ for $c$ with empty resource. Thus, $\mu$ cannot be induced by an optimal profile of cutoffs. Contradiction.
    
    Now we show that $(s,c',r)\in \mu$, where $\{c,c'\}\subseteq C_r$, i.e., student $s$ is admitted to another college from the same region with a unit of $r$. Note that all resource $r$'s units should be distributed in $C_r$. Suppose otherwise, i.e., college $c$ has an empty seat and can distribute a unit of $r$. Since $\mu$ is induced by an optimal profile of cutoffs, there is no blocking contract with college $c$ and an empty recourse $r_0$. Therefore, there should exist a seat-blocking contract $x$ containing $c$ and $r$ (it is either $(s,c,r)$ or some contract with higher rank at $c$), such that, if we set both $\Cutoffs[0][c]$ and $\Cutoffs[r][c]$ to be equal to the rank of student from $x$, we induce a feasible matching different from $\mu$, thus, $\mu$ cannot be induced by an optimal profile of cutoffs. Contradiction. As a result, for $(s,c,r)$ to be seat-blocking, $s$ should be admitted with $r$ to some college in the same region for $r$ as $c$. 
    
    Finally, we show that $x\neq(s,c,r)$, which implies that $(s,c,r)$ is dominated. Suppose that $x=(s,c,r)$, then if we set both $\Cutoffs[0][c]$ and $\Cutoffs[r][c]$ to be equal to the rank of $s$, we induce a feasible matching different from $\mu$, thus, $\mu$ cannot be induced by an optimal profile of cutoffs. Contradiction.

    Now suppose that there is a resource-blocking contract $(s,c,r)$ under $\mu$. Thus, $r\neq r_0$, $(s,c,r')\in \mu$ with $r'\neq r$, and there is a free unit of $r$ that $c$ can distribute. Suppose that $(s,c,r)$ is undominated, i.e., there is no contract $(s',c,r)$, such that, $s'\succ_c s$ and $(c,r)\succ_{s'} \mu_{s'}$. This implies that $\Cutoffs[r][c]$ is equal to the rank of $s$ minus one. Thus, matching $\mu(\CutoffsMatrix+ \mathbf{1}_{r}^c)=(\mu(\CutoffsMatrix)\backslash (s,c,r'))\cup (s,c,r)$ is feasible. Contradiction.
    
    Thus, $\mu$ is direct-envy stable, i.e., any optimal cutoff profile induces a direct-envy stable matching.

    Now we switch to the second statement. Consider a direct-envy stable matching $\mu$. Construct a profile of cutoffs $\CutoffsMatrix$, such that for each pair $c\in \Colleges$ and $r\in \Resources_0$, the cutoff $\Cutoffs[r][c]$ either equals to the rank of $s$ minus one, where $x_r^c\equiv (s,c,r)$ is a contract with the highest ranked student $s$ by $c$, such that $(c,r) \succ_s \mu_s$, or is maximal, if no such contract exists. We need to show that, first, $\CutoffsMatrix$ induces $\mu$, and, second, $\CutoffsMatrix$ is optimal.

    Take any contract $(s,c,r)\in \mu$. By construction, $\Cutoffs[r][c]\leq$ the rank of $s$, since otherwise there will be a direct envy-blocking contract $x_r^c$. So, contract $(s,c,r)$ is included into $\mu(\CutoffsMatrix)$, since it is the best contract for $s$ that she can choose from. Thus, $\CutoffsMatrix$ induces $\mu$.

    Take any not maximal cutoff $\Cutoffs[r][c]$ with a corresponding contract $x_r^c = (s,c,r)$, where the rank of $s$ is equal to $\Cutoffs[r][c]+1$. The constructed cutoff profile is not optimal if a matching $(\mu(\CutoffsMatrix)\backslash \mu(\CutoffsMatrix)_s)\cup x_r^c$ is feasible. It may happen in one of two cases:
    \begin{enumerate}
        \item if all seats of $c$ are taken, but a unit of $r$ for $c$ is free, i.e., some $(s,c,r')\in \mu$, which implies that $(s,c,r)$ is undominated resource-blocking contract. Contradiction.

        \item if there is an empty seat at $c$, i.e., $(s,c,r)$ is undominated seat-blocking contract. Contradiction.
    \end{enumerate}

    Thus, $\CutoffsMatrix$ is the only optimal cutoff profile that induces $\mu$. As a result, any direct-envy stable matching can be induced by unique optimal cutoff profile.
\qed\medskip

\noindent\textit{Proof of \Cref{prop:IRC_IMC_IDC_des}}. By construction, the resulting matching $\mu$ in all three mechanisms is induced by an optimal profile of cutoffs. Thus, by \Cref{prop:optimal_cuts_des}, $\mu$ is direct-envy stable.

Since there are at most $N^2$ cutoffs, each of which can be updated at most $N$ times, the overall computational complexity of both mechanisms is $\mathcal{O}(N^3)$.
\qed\medskip

\noindent\textit{Proof of \Cref{prop:properties_IUC}}. IUC is envy-free by construction. It cannot produce any envy-blocking or resource-blocking contracts, because each college has all its cutoffs being equal.

Since there are at most $N$ cutoffs, each of which can be updated at most $N$ times, the overall computational complexity of IUC is $\mathcal{O}(N^2)$.
\qed\medskip

\noindent\textit{Proof of \Cref{prop:cutoff_mechs_not_strategy-proof}}. 
    Consider a market with two students, two colleges, and one resource. Corresponding region contains both colleges. Quotas are $q_{c_1}=q_{c_2}=q_r=1$. Preferences are:

    \begin{table}[H]
    \centering
    \begin{tabular}{l c c}
    $\succ_{s_1}: (c_2,r),(c_1,r),\emptyset$ &  & $\succ_{c_1}: s_1,s_2$\\
    $\succ_{s_2}: (c_2,r),\emptyset$ &  & $\succ_{c_2}: s_2,s_1$
    \end{tabular}
    \end{table}
    
    Suppose, first we choose $c_1$ for IUC, IDC, or IMC, or $(c_1,r)$ for IRC. Thus, all for mechanisms yield $\{(s_1,c_1,r)\}$. Note that there is only one stable matching $\{(s_2,c_2,r)\}$. Therefore, IUC, IMC, IRC, and IDC do not always find a stable matching whenever one exists.
    
    Now, consider a market with two students, two colleges, and only empty resource. Quotas are $q_{c_1}=q_{c_2}=1$. Preferences are:
    
    \begin{table}[H]
    \centering
    \begin{tabular}{l c c}
    $\succ_{s_1}: (c_2,r_0),(c_1,r_0),\emptyset$ &  & $\succ_{c_1}: s_1,s_2$\\
    $\succ_{s_2}: (c_1,r_0),(c_2,r_0),\emptyset$ &  & $\succ_{c_2}: s_2,s_1$
    \end{tabular}
    \end{table}
    
    Suppose, first we choose $c_1$ for IUC, IDC, or IMC or $(c_1,r_0)$ for IRC; then we choose $c_2$ for IUC, IDC, or IMC or $(c_2,r_0)$ for IRC; then again $c_1$ for IMC or $(c_1,r_0)$ for IRC; then we choose $c_2$ for IUC, IDC, or IMC or $(c_2,r_0)$ for IRC. Thus, all mechanisms yield $\{(s_1,c_1,r_0),(s_2,c_2,r_0)\}$.
    
    Now, suppose that student $s_1$ lies about her preferences: $\succ'_{s_1}:(c_2,r_0),\emptyset$. Thus, now all mechanisms yield $\{(s_1,c_2,r_0),(s_2,c_1,r_0)\}\succ_{s_1}\{(s_1,c_1,r_0),(s_2,c_2,r_0)\}$. Thus, IUC, IMC, IRC, and IDC are not strategy-proof for students.\footnote{For this case, IUC, IMC, IRC, and IDC act as college-proposing deferred acceptance, which is not strategy-proof for students.}
\qed\medskip

\noindent\textit{Proof of \Cref{prop:cutoffs_mechanisms_are_stable_if_no_resources}}. By construction, IMC, IRC, IDC and IUC mechanisms become college-proposing deferred acceptance (CDA) under no non-empty resources. Thus, first, they produce the same matching for the same market, and, second, all properties of CDA transfer to IMC, IRC, IDC and IUC, e.g., stability of the resulting matching (since stability from \Cref{def:stability} becomes classical \citet{gale62} stability under no non-empty resources).
\qed\medskip

\noindent\textit{Proof of \Cref{prop:properties_SD}}. RSD and CSD are Pareto-efficient by construction. In turn, Pareto-efficiency implies non-wastefulness (absence of waste-blocking and resource-blocking contracts). RSD is strategy-proof by construction. CSD is not strategy-proof as students may affect their chances of being chosen earlier by putting contracts with high college rank higher in her own ranking.

Since there are at most $N$ students to call, the overall computational complexity of both mechanisms is $\mathcal{O}(N)$.
\qed\medskip

\noindent\textit{Proof of \Cref{prop:CSD_is_stable_under_colleges_alignment}}. Consider the common preferences of colleges $\succ$. By construction, CSD call order of students is exactly $\succ$, i.e., the best student chooses a contract first, second best chooses a contract second, etc. Thus, no student may affect this order and, moreover, each time any student is called to choose, her best action is to choose the best possible contract according to her true preferences, i.e., CSD is strategy-proof.

By \Cref{prop:properties_SD}, CDS is non-wasteful and is resource-efficient. Moreover, from above we know that call order is exactly $\succ$, therefore, by \Cref{def:envy_freeness}, there may not be envy in the final matching. As a result, by \Cref{def:stability}, CSD is stable.
\qed

\section{Implementation Details}\label{app:implementation_details}

\subsection{Preferences}

Our numerical results consider three kind of alignments on preferences. For each kind of market, we first draw preferences without unacceptable contracts with the following rules, and then randomly erase contracts from each agent's preferences to simulate unacceptable contracts. Due to this, even under the full alignment regime, agents can present different preferences.

\begin{itemize}
    \item[1.] \textbf{No alignment}: College preferences are unconstrained, and student's preferences are such that $(c,r)\succ_s (c,r_0)$ for each student $s$, college $c$, and resource $r$.
    \item[2.] \textbf{Semi-alignment}: Preferences are sampled from a parametric distribution where certain contracts are more likely to be ranked higher. We explain them in detail below.
    \item[3.] \textbf{Full alignment}: Agents from the same side agree on the preferences.
    \begin{itemize}
        \item Student full-alignment: Colleges' preferences are unconstrained, while students' preferences are aligned over colleges and resources separately, i.e., for any student $s\in\Students$:
    \begin{itemize}
        \item for any resource $r \in \Resources_0$, $(c_{|{\Colleges}|},r) \succ_s ... \succ_s (c_2,r) \succ_s (c_1,r)$, and
        \item for any college $c \in \Colleges$, 
        $(c,r_{|{\Resources}|}) \succ_s ... \succ_s (c,r_2) \succ_s (c,r_1) \succ_s (c,r_0)$.
    \end{itemize}     
    \item College full-alignment: All colleges agree on the preferences over students, i.e., for any college $c\in\Colleges$, $s_{|{\Students}|} \succ_c ... \succ_c s_2 \succ_c s_1$, while students have no-alignment preferences as above.
    \item Students and College full-alignment: Both students and colleges have aligned preferences as described above.
    \end{itemize}
\end{itemize}

Regarding semi-aligned preferences, we only implemented it for students, as real-life scenarios rarely show colleges semi-aligned. \Cref{alg:aligned_preferences} details how to sample a semi-aligned preference ordering for a given student.

\begin{algorithm}[ht]
\caption{Aligned preferences for student $s$}
\label{alg:aligned_preferences}
Pref($s$) $\longleftarrow [\ ]$

Contracts $\longleftarrow \{c_{\mathcal{C}},r_{\mathcal{R}}\}$ (the most preferred contract)

\While{Contracts is non-empty}{

Randomly select a contract $(c,r)$ from Contracts

Contracts $\longleftarrow$ Contracts $\setminus\ (c,r)$
        
Pref($s$) $\longleftarrow$  Pref($s$) $\cup\ (c,r)$

\For{each contract in $\mathcal{C}\times\mathcal{R}_0$ immediately lower in rank than $(c,r)$}{
Add the lower-ranked contract to Contracts
}
}
Truncate Pref$(s)$ at a random position
\end{algorithm}

Colleges are endowed with qualities (College $1$ has the lowest quality, College $C$ has the highest quality). Given a student $s$, with a (partially created, initially empty) preference ordering Pref$(s)$, we sample a college $c$ from $\mathcal{C}$ with probability
\begin{align*}
    p_c := \frac{\text{quality}(c)}{\sum_{c'\in\mathcal{C}} \text{quality}(c')},
\end{align*}
that is, colleges with higher quality are more likely to be ranked higher, and randomly choose a non-empty contract $r$ such that $(c,r)$ does not already belong to Pref$(s)$, and add it to Pref$(s)$. If for a college $c$ all contracts $(c,r)$ with non-empty resources $r$ have been already included in Pref$(s)$, then we add $(c,r_0)$.

\subsection{Balanced and Unbalanced Markets}

We study the stability of the mechanisms introduced in \Cref{section:mechanisms} empirically under both balance and unbalanced markets. In balanced markets, the total number of seats and the total number of each kind of resource equal the total number of students. For unbalanced markets, we considered two different ways:
\begin{itemize}
    \item[$\bullet$] \textbf{Colleges Up/Down}: The sum of colleges' capacities is higher/lower than the total number of students,
    \item[$\bullet$] \textbf{Resources Up/Down}: The sum of resources' quotas is higher/lower than the total number of students.
\end{itemize}
For the Up markets we consider numbers equal to twice the number of students, e.g., for a Colleges Up market, the sum of capacities of all colleges sum up 2 times the number of available students on the market. For the Down versions, we consider numbers equal to the half of the available students. By mixing both, we simulate four different kinds of unbalanced markets. 

\section{Numerical Results}\label{app:numerical_results}

In this section we show complementary numerical experiments. The main conclusions are the same as in the main body: whenever colleges present alignment on preferences, the best mechanism is CSD. In the rest of the cases, the best one is IMC.

\begin{table}[H]
\centering
\footnotesize{
\caption{Average numbers of blocking contracts over $100$ balanced markets with $100$ students, $10$ colleges, and $1$ resource $\Resources_0 = \{r_0\}$, i.e., college admissions à la \citet{gale62}.}
\begin{tabular}{l|cccccc}
\toprule
Alignment & Mechanism & Resource & Seat   & Direct-Envy   & Indirect-Envy       & Total\\
\hline
 &IRC & 0.0±0.0     & 0.0±0.0 & 0.0±0.0       & 0.0±0.0 & 0.0±0.0     \\
 &IMC & 0.0±0.0     & 0.0±0.0 & 0.0±0.0       & 0.0±0.0 & 0.0±0.0     \\
No &IDC & 0.0±0.0     & 0.0±0.0 & 0.0±0.0       & 0.0±0.0 & 0.0±0.0   \\
Alignment &IUC & 0.0±0.0     & 0.0±0.0 & 0.0±0.0       & 0.0±0.0 & 0.0±0.0\\
 &RSD & 0.0±0.0     & 0.0±0.0 & 29.12±8.341   & 0.0±0.0 & 29.12±8.341 \\
 &CSD & 0.0±0.0     & 0.0±0.0 & 4.84±3.877    & 0.0±0.0 & 4.84±3.877  \\
\toprule
&IRC & 0.0±0.0     & 0.0±0.0 & 0.0±0.0       & 0.0±0.0 & 0.0±0.0      \\
Students &IMC & 0.0±0.0     & 0.0±0.0 & 0.0±0.0       & 0.0±0.0 & 0.0±0.0      \\
Full &IDC & 0.0±0.0     & 0.0±0.0 & 0.0±0.0       & 0.0±0.0 & 0.0±0.0      \\
Alignment &IUC & 0.0±0.0     & 0.0±0.0 & 0.0±0.0       & 0.0±0.0 & 0.0±0.0      \\
&RSD & 0.0±0.0     & 0.0±0.0 & 151.35±22.17  & 0.0±0.0 & 151.35±22.17 \\
&CSD & 0.0±0.0     & 0.0±0.0 & 4.48±3.999    & 0.0±0.0 & 4.48±3.999   \\
\toprule
  &IRC & 0.0±0.0     & 0.0±0.0 & 0.0±0.0       & 0.0±0.0 & 0.0±0.0     \\
Colleges  &IMC & 0.0±0.0     & 0.0±0.0 & 0.0±0.0       & 0.0±0.0 & 0.0±0.0     \\
Full  &IDC & 0.0±0.0     & 0.0±0.0 & 0.0±0.0       & 0.0±0.0 & 0.0±0.0     \\
Alignment  &IUC & 0.0±0.0     & 0.0±0.0 & 0.0±0.0       & 0.0±0.0 & 0.0±0.0     \\
  &RSD & 0.0±0.0     & 0.0±0.0 & 30.41±9.868   & 0.0±0.0 & 30.41±9.868 \\
  &CSD & 0.0±0.0     & 0.0±0.0 & 0.0±0.0       & 0.0±0.0 & 0.0±0.0     \\
\toprule
&  IRC & 0.0±0.0     & 0.0±0.0 & 0.0±0.0       & 0.0±0.0 & 0.0±0.0 \\
Students&  IMC & 0.0±0.0     & 0.0±0.0 & 0.0±0.0       & 0.0±0.0 & 0.0±0.0       \\
and&  IDC & 0.0±0.0     & 0.0±0.0 & 0.0±0.0       & 0.0±0.0 & 0.0±0.0\\
Colleges&  IUC & 0.0±0.0     & 0.0±0.0 & 0.0±0.0       & 0.0±0.0 & 0.0±0.0       \\
Full&  RSD & 0.0±0.0     & 0.0±0.0 & 154.57±22.078 & 0.0±0.0 & 154.57±22.078 \\
Alignment&  CSD & 0.0±0.0     & 0.0±0.0 & 0.0±0.0       & 0.0±0.0 & 0.0±0.0       \\
\toprule
\end{tabular}
\label{tab:blocking_contracts_no_resources}
}
\end{table}

\subsection{No Resources}

\Cref{tab:blocking_contracts_no_resources} illustrates \Cref{prop:cutoffs_mechanisms_are_stable_if_no_resources} that all cutoffs mechanisms become Deferred-Acceptance when no non-empty resources are available on the market.
\vspace{-0.5cm}

\subsection{Experiments on Balanced Markets}

We present two extra experiments on balanced markets for 2 and 10 resources, illustrated in \Cref{tab:blocking_contracts_two_resources}.

\subsection{Experiments on Unbalanced Markets}\label{Appendix:Unbalanced markets}

To complement the experiments on unbalanced markets in \Cref{sec:empirics}, we show the performance of the mechanisms in unbalanced markets with two resources when colleges present full alignment on preferences.

Unlike the results observed for balanced horizontal markets with two resources, now all mechanisms produce blocking contracts. Waste is clearly affected by colleges capacities, as for colleges down markets the mechanisms manage to occupy most of their seats. The serial dictatorship mechanisms seem to benefit from colleges up markets as they reduce their indirect envy to zero. Indeed, as colleges have larger quotas, the serial dictatorship mechanisms can allocate to the students to their most preferred colleges. IMC remains the most stable mechanism in average, although for Colleges Up market, its advantage decreases considerably with respect to IRC and CSD.

\begin{table}[H]
\centering
\footnotesize{
\caption{Average numbers of blocking contracts over $100$ balanced markets with $100$ students and $10$ colleges.}
\begin{tabular}{l|cccccc}
\toprule
\multicolumn{7}{c}{\textbf{Two resources ($\Resources_0 = \{r_0,r_1\}$)}}\\
\toprule
Alignment & Mechanism & Resource & Seat   & Direct-Envy   & Indirect-Envy       & Total   \\
\hline
  & IRC & 1.34±1.298  & 0.0±0.0 & 0.0±0.0       & 5.66±4.186  & 7.0±4.879   \\
  & IMC & 0.0±0.0     & 0.0±0.0 & 0.0±0.0       & 0.0±0.0     & \textbf{0.0±0.0}     \\
No  & IDC & 3.62±2.253  & 0.0±0.0 & 0.0±0.0       & 16.75±7.676 & 20.37±9.251 \\
Alignment  & IUC & 0.0±0.0     & 0.0±0.0 & 0.0±0.0       & 0.0±0.0     & \textbf{0.0±0.0}     \\
  & RSD & 0.0±0.0     & 0.0±0.0 & 41.1±13.238   & 0.7±1.162   & 41.8±13.533 \\
  & CSD & 0.0±0.0     & 0.0±0.0 & 8.74±6.144    & 0.17±0.549  & 8.91±6.337  \\
\toprule
 &IRC & 2.23±1.923  & 0.0±0.0 & 0.0±0.0       & 4.63±3.322   & 6.86±4.741 \\
Students &IMC & 0.0±0.0     & 0.0±0.0 & 0.0±0.0       & 0.0±0.0      & \textbf{0.0±0.0}    \\
Full &IDC & 10.96±5.658 & 0.0±0.0 & 0.0±0.0       & 21.85±11.535 & 32.81±16.321 \\
Alignment &IUC & 0.0±0.0     & 0.0±0.0 & 0.0±0.0       & 0.0±0.0      & \textbf{0.0±0.0}    \\
 &RSD & 0.0±0.0     & 0.0±0.0 & 280.65±40.097 & 3.24±5.318   & 283.89±41.46 \\
 &CSD & 0.0±0.0     & 0.0±0.0 & 8.09±5.07     & 0.3±0.806    & 8.39±5.291   \\
\toprule
  &IRC & 1.18±1.203  & 0.0±0.0 & 0.0±0.0       & 3.42±2.786  & 4.6±3.701    \\
Colleges  &IMC & 0.0±0.0     & 0.0±0.0 & 0.0±0.0       & 0.0±0.0     & \textbf{0.0±0.0}      \\
Full &IDC & 3.93±2.747  & 0.0±0.0 & 0.0±0.0       & 14.21±8.644 & 18.14±10.797 \\
Alignment &IUC & 0.0±0.0     & 0.0±0.0 & 0.0±0.0       & 0.0±0.0     & \textbf{0.0±0.0}      \\
  &RSD & 0.0±0.0     & 0.0±0.0 & 40.55±13.776  & 0.55±1.117  & 41.1±13.99   \\
  &CSD & 0.0±0.0     & 0.0±0.0 & 0.0±0.0       & 0.0±0.0     & \textbf{0.0±0.0}      \\
\toprule
Students  &IRC & 4.92±2.887  & 0.0±0.0 & 0.0±0.0       & 12.36±6.552  & 17.28±8.594            \\
and &IMC & 0.0±0.0     & 0.0±0.0 & 0.0±0.0       & 0.0±0.0      & \textbf{0.0±0.0}\\
Colleges  &IDC & 13.59±6.322 & 0.0±0.0 & 0.0±0.0       & 36.89±15.891 & 50.48±21.072            \\
Full  &IUC & 0.0±0.0     & 0.0±0.0 & 0.0±0.0       & 0.0±0.0      & \textbf{0.0±0.0}       \\
Alignment  &RSD & 0.0±0.0     & 0.0±0.0 & 272.73±35.013 & 3.39±5.065   & 276.12±36.115           \\
  &CSD & 0.0±0.0     & 0.0±0.0 & 0.0±0.0       & 0.0±0.0      & \textbf{0.0±0.0}   \\
\toprule
\multicolumn{7}{c}{\textbf{Ten resources ($\Resources_0 = \{r_0,r_1,...,r_9\}$)}}\\
\toprule
&IRC & 0.41±0.736  & 0.66±0.951     & 0.0±0.0       & 9.46±6.034  & 10.53±6.579 \\
&IMC & 0.07±0.292  & 0.8±1.241      & 0.0±0.0       & 3.26±2.928  & \textbf{4.13±3.303}     \\
No &IDC & 0.63±0.956  & 0.92±1.309     & 0.0±0.0       & 18.91±9.061 & 20.46±9.497      \\
Alignment&IUC & 0.0±0.0     & 752.08±387.581 & 0.0±0.0       & 0.0±0.0     & 752.08±387.581 \\
&RSD & 0.0±0.0     & 0.0±0.0        & 75.48±24.962  & 16.31±8.423 & 91.79±29.37      \\
&CSD & 0.0±0.0     & 0.0±0.0        & 28.66±13.258  & 6.94±4.308  & 35.6±15.77     \\
\toprule
&IRC & 1.49±1.921  & 2.16±2.962     & 0.0±0.0       & 18.95±6.331   & 22.6±7.588        \\
Students  &IMC & 0.03±0.171  & 2.37±3.186     & 0.0±0.0       & 9.67±3.829    & \textbf{12.07±4.704}    \\
Full  &IDC & 8.66±8.565  & 3.71±5.206     & 0.0±0.0       & 51.06±22.595  & 63.43±29.689     \\
Alignment  &IUC & 0.0±0.0     & 3344.0±488.829 & 0.0±0.0       & 0.0±0.0       & 3344.0±488.829  \\
  &RSD & 0.0±0.0     & 0.0±0.0        & 623.57±69.089 & 261.15±58.571 & 884.72±89.696  \\
  &CSD & 0.0±0.0     & 0.0±0.0        & 28.74±9.948   & 10.56±5.447   & 39.3±13.706      \\
\toprule
&IRC & 1.05±2.559  & 0.11±0.343     & 0.0±0.0       & 7.28±8.223  & 8.44±10.58     \\
Colleges  &IMC & 0.0±0.0     & 0.01±0.099     & 0.0±0.0       & 0.01±0.099  & \textbf{0.02±0.14}          \\
Full  &IDC & 1.26±1.906  & 0.75±1.081     & 0.0±0.0       & 17.3±10.395 & 19.31±11.836       \\
Alignment  &IUC & 0.0±0.0     & 715.97±387.724 & 0.0±0.0       & 0.0±0.0     & 715.97±387.724 \\
  &RSD & 0.0±0.0     & 0.0±0.0        & 68.85±27.591  & 14.37±6.583 & 83.22±30.454   \\
  &CSD & 0.0±0.0     & 0.0±0.0        & 0.0±0.0       & 0.0±0.0     & \textbf{0.0±0.0}      \\
\toprule
Students&IRC & 3.78±3.463  & 0.53±1.261     & 0.0±0.0       & 38.91±13.23   & 43.22±15.606  \\
and & IMC & 0.0±0.0     & 0.0±0.0        & 0.0±0.0       & 0.02±0.14     & \textbf{0.02±0.14}           \\
Colleges  &IDC & 9.97±12.438 & 2.01±3.369     & 0.0±0.0       & 103.62±38.303 & 115.6±48.831   \\
Full  &IUC & 0.0±0.0     & 3382.99±487.98 & 0.0±0.0       & 0.0±0.0       & 3382.99±487.98 \\
Alignment  &RSD & 0.0±0.0     & 0.0±0.0        & 622.06±81.873 & 256.43±48.975 & 878.49±99.351  \\
  &CSD & 0.0±0.0     & 0.0±0.0        & 0.0±0.0       & 0.0±0.0       & \textbf{0.0±0.0   } \\
\toprule
\end{tabular}
\label{tab:blocking_contracts_two_resources}
}
\end{table}

For Student full aligned markets and Student-College full aligned markets, the overall performance of the mechanisms remains quite similar to the ones in \Cref{tab:blocking_contracts_two_resources_unbalanced_horizontal_market}, so we prefer to omit them. For College full aligned markets, however, IMC seems to become fully stable, as illustrated in \Cref{tab:blocking_contracts_two_resources_unbalanced_college_vertical_market}. This was also observed in experiments with five resources (results can be found with the code). \Cref{tab:blocking_contracts_IMC_five_resources_unbalanced_markets} presents some of the few cases where IMC produces blocking contracts. 

However, unlike \citet{ashlagi17} model, it is not clear that unbalanced markets help in achieving more stable matchings under RRC.

\begin{table}[t]
\centering
\footnotesize{
\caption{Average numbers of blocking contracts produced over $100$ unbalanced markets with $100$ students, $10$ colleges, $2$ resources $\Resources_0 = \{r_0,r_1\}$, and full college alignment.}
\begin{tabular}{l|cccccc}
\toprule
Market & Mechanism & Resource & Seat   & Direct-Envy   & Indirect-Envy       & Total \\
\hline
&IRC & 0.01±0.099  & 0.58±1.022     & 0.0±0.0       & 0.04±0.196      & 0.63±1.083  \\
Colleges Up  &IMC & 0.0±0.0     & 0.0±0.0        & 0.0±0.0       & 0.0±0.0         & \textbf{0.0±0.0}   \\
  &IDC & 0.02±0.14   & 3.16±3.972     & 0.0±0.0       & 0.08±0.337      & 3.26±4.034     \\
Resources Up &IUC & 0.0±0.0     & 135.02±170.054 & 0.0±0.0       & 0.0±0.0         & 135.02±170.054        \\
  &RSD & 0.0±0.0     & 0.0±0.0        & 39.0±42.709   & 0.0±0.0         & 39.0±42.709              \\
  &CSD & 0.0±0.0     & 0.0±0.0        & 0.0±0.0       & 0.0±0.0         & \textbf{0.0±0.0}      \\
  \toprule
&IRC & 1.8±1.8     & 0.0±0.0       & 0.0±0.0       & 4.21±3.819      & 6.01±5.354             \\
Colleges Down  &IMC & 0.0±0.0     & 0.0±0.0       & 0.0±0.0       & 0.0±0.0         & \textbf{0.0±0.0}          \\
  &IDC & 4.8±3.644   & 0.0±0.0       & 0.0±0.0       & 14.9±10.396     & 19.7±13.65    \\
Resources Up  &IUC & 0.0±0.0     & 76.55±140.236 & 0.0±0.0       & 0.0±0.0         & 76.55±140.236          \\
  &RSD & 0.0±0.0     & 0.0±0.0       & 396.19±56.79  & 6.52±9.975      & 402.71±57.696           \\
  &CSD & 0.0±0.0     & 0.0±0.0       & 0.0±0.0       & 0.0±0.0         & \textbf{0.0±0.0}              \\
\toprule
&IRC & 0.0±0.0     & 1.05±1.571    & 0.0±0.0       & 0.03±0.171      & 1.08±1.566       \\
Colleges Down & IMC & 0.0±0.0     & 0.0±0.0       & 0.0±0.0       & 0.0±0.0         & \textbf{0.0±0.0}               \\
 & IDC & 0.0±0.0     & 0.0±85.326    & 0.0±0.0       & 0.0±0.0         & 0.0±85.326            \\
Resources Up & IUC & 0.0±0.0     & 350.68±85.326 & 0.0±0.0       & 0.0±0.0         & 350.68±85.326           \\
 & RSD & 0.0±0.0     & 0.0±0.0       & 86.5±28.483   & 0.0±0.0         & 86.5±28.483             \\
 & CSD & 0.0±0.0     & 0.0±0.0       & 0.0±0.0       & 0.0±0.0         & \textbf{0.0±0.0}      \\
\toprule
&IRC & 0.54±1.153  & 0.0±0.0       & 0.0±0.0       & 1.52±1.808      & 2.06±2.742              \\
Colleges Down & IMC & 0.0±0.0     & 0.0±0.0       & 0.0±0.0       & 0.0±0.0         & \textbf{0.0±0.0}          \\
 & IDC & 0.0±0.0     & 0.0±138.91    & 0.0±0.0       & 0.0±0.0         & 0.0±138.91              \\
Resources Down & IUC & 0.0±0.0     & 236.81±138.91 & 0.0±0.0       & 0.0±0.0         & 236.81±138.91           \\
 & RSD & 0.0±0.0     & 0.0±0.0       & 372.39±86.036 & 1.55±5.324      & 373.94±87.026           \\
 & CSD & 0.0±0.0     & 0.0±0.0       & 0.0±0.0       & 0.0±0.0         & \textbf{0.0±0.0}              \\
  \toprule
\end{tabular}
\label{tab:blocking_contracts_two_resources_unbalanced_college_vertical_market}
}
\end{table}

\begin{table}[t]
\centering
\footnotesize{
\caption{Average numbers of blocking contracts produced by the IMC mechanism over $100$ unbalanced markets with $100$ students, $10$ colleges, and $5$ resources.}
\begin{tabular}{cccccc}
\toprule
\multicolumn{6}{c}{\textbf{Increasing Minimal Cutoffs}}\\
\toprule
\multicolumn{6}{c}{\textbf{No Alignment - Five resources ($\Resources_0 = \{r_0,r_1,...,r_4\}$)}}\\
\toprule
Market & Resource & Seat   & Direct-Envy   & Indirect-Envy       & Total    \\
\hline
Colleges Down - Resources Up & 0.0±0.0     & 0.0±0.0        & 0.0±0.0       & 0.36±0.855     & 0.36±0.855     \\
\toprule
\multicolumn{6}{c}{\textbf{College full-aligned Markets - Five resources ($\Resources_0 = \{r_0,r_1,...,r_4\}$)}}\\
\toprule
Market & Resource & Seat   & Direct-Envy   & Indirect-Envy       & Total    \\
\hline
Colleges Up - Resources Up & 0.0±0.0     & 0.01±0.099     & 0.0±0.0        & 0.0±0.0   & 0.01±0.099     \\
\hline
Colleges Up - Resources Down & 0.0±0.0     & 0.06±0.237     & 0.0±0.0       & 0.0±0.0   & 0.06±0.237     \\
\hline
Colleges Down - Resources Down & 0.0±0.0     & 0.0±0.0        & 0.0±0.0        & 0.03±0.171   & 0.03±0.171     \\
\toprule
\multicolumn{6}{c}{\textbf{Student full-aligned Markets - Five resources ($\Resources_0 = \{r_0,r_1,...,r_4\}$)}}\\
\toprule
Market & Resource & Seat   & Direct-Envy   & Indirect-Envy       & Total    \\
\hline
Colleges Down - Resources Up & 0.0±0.0     & 0.0±0.0        & 0.0±0.0        & 0.31±0.643     & 0.31±0.643      \\
\toprule
\end{tabular}
\label{tab:blocking_contracts_IMC_five_resources_unbalanced_markets}
}
\end{table}

\subsection{Experiments on Students partially-aligned markets}

\Cref{tab:rebuttals} summarizes the performances of the six mechanisms in balanced and unbalanced markets with partially-aligned students, no aligned colleges preferences, and five resources. We observe that IMC, once again, achieves the best global performance.

\begin{table}[t]
\centering
\footnotesize{
\caption{Average numbers of blocking contracts produced by five mechanisms over $100$ Student partial-aligned markets with $100$ students, $10$ colleges, and $5$ resources $\mathcal{R}_0 = \{r_0,r_1,r_2,r_3,r_4\}$.}
\begin{tabular}{cccccc}
\toprule
\multicolumn{6}{c}{\textbf{Student partial-alignment Markets}}\\
\toprule
\multicolumn{6}{c}{\textbf{Balanced}} \\
\hline
& Resource   & Waste   & Direct-Envy   & Indirect-Envy       & Total   \\
\hline
IRC & 1.73±2.315  & 1.03±1.758     & 0.0±0.0        & 13.06±10.469 & 15.82±11.646     \\
  IMC & 0.06±0.369  & 0.89±1.28      & 0.0±0.0        & 1.98±2.709   & \textbf{2.93±3.468}  \\
  IDC & 4.8±5.255   & 1.61±2.825     & 0.0±0.0        & 36.06±28.787 & 42.47±32.336  \\
  IUC & 0.0±0.0     & 354.27±321.113 & 0.0±0.0        & 0.0±0.0      & 354.27±321.113 \\
  RSD & 0.0±0.0     & 0.0±0.0        & 213.12±110.211 & 23.66±22.408 & 236.78±100.5   \\
  CSD & 0.0±0.0     & 0.0±0.0        & 47.28±23.946   & 7.15±8.476   & 54.43±26.635   \\
\hline
\multicolumn{6}{c}{\textbf{Unbalanced: Colleges Up - Resources Up}} \\
\hline
IRC & 0.02±0.14   & 6.36±7.674    & 0.0±0.0       & 0.98±1.435 & 7.36±7.829      \\
  IMC & 0.0±0.0     & 5.98±7.511    & 0.0±0.0       & 0.53±1.1   & \textbf{6.51±7.988 }       \\
  IDC & 0.16±0.418  & 7.03±8.295    & 0.0±0.0       & 2.01±2.9   & 9.2±8.246     \\
  IUC & 0.0±0.0     & 386.03±355.59 & 0.0±0.0       & 0.0±0.0    & 386.03±355.59   \\
  RSD & 0.0±0.0     & 0.0±0.0       & 92.32±94.469  & 0.3±0.755  & 92.62±94.27    \\
  CSD & 0.0±0.0     & 0.0±0.0       & 16.34±20.63   & 0.03±0.171 & 16.37±20.624    \\
\hline
\multicolumn{6}{c}{\textbf{Unbalanced: Colleges Down - Resources Up}} \\
\hline
IRC & 3.3±3.094   & 0.06±0.369     & 0.0±0.0        & 7.95±5.601     & 11.31±8.143      \\
  IMC & 0.0±0.0     & 0.1±0.458      & 0.0±0.0        & 0.2±0.663      & \textbf{0.3±0.964} \\
  IDC & 14.89±9.88  & 0.02±0.199     & 0.0±0.0        & 35.22±27.379   & 50.13±34.605  \\
  IUC & 0.0±0.0     & 251.73±309.947 & 0.0±0.0        & 0.0±0.0        & 251.73±309.947\\
  RSD & 0.0±0.0     & 0.0±0.0        & 610.62±112.108 & 268.34±177.467 & 878.96±252.025\\
  CSD & 0.0±0.0     & 0.0±0.0        & 94.07±35.391   & 49.37±36.149   & 143.44±66.788 \\
\hline
\multicolumn{6}{c}{\textbf{Unbalanced: Colleges Up - Resources Down}} \\
\hline
IRC & 0.0±0.0     & 9.21±5.611    & 0.0±0.0       & 1.24±1.594 & 10.45±5.752   \\
  IMC & 0.0±0.0     & 9.16±5.617    & 0.0±0.0       & 1.2±1.428  & \textbf{10.36±5.698 }\\
  IDC & 0.0±0.0     & 10.32±6.005   & 0.0±0.0       & 1.28±1.393 & 11.6±6.065  \\
  IUC & 0.0±0.0     & 731.8±237.902 & 0.0±0.0       & 0.0±0.0    & 731.8±237.902\\
  RSD & 0.0±0.0     & 0.0±0.0       & 191.89±63.694 & 0.0±0.0    & 191.89±63.694  \\
  CSD & 0.0±0.0     & 0.0±0.0       & 18.8±15.544   & 0.0±0.0    & 18.8±15.544   \\
  \hline
\multicolumn{6}{c}{\textbf{Unbalanced: Colleges Down - Resources Down}} \\
\hline
IRC & 0.67±1.619  & 0.18±0.517     & 0.0±0.0        & 2.05±2.414   & 2.9±3.78           \\
  IMC & 0.0±0.0     & 0.23±0.526     & 0.0±0.0        & 0.73±0.798   & \textbf{0.96±0.999}         \\
  IDC & 4.78±6.691  & 0.13±0.365     & 0.0±0.0        & 10.91±14.09  & 15.82±20.092  \\
  IUC & 0.0±0.0     & 701.06±247.746 & 0.0±0.0        & 0.0±0.0      & 701.06±247.746\\
  RSD & 0.0±0.0     & 0.0±0.0        & 539.57±115.852 & 14.52±25.431 & 554.09±128.86 \\
  CSD & 0.0±0.0     & 0.0±0.0        & 51.51±30.78    & 2.27±4.345   & 53.78±33.362  \\
\toprule
\end{tabular}
\label{tab:rebuttals}
}
\end{table}

\section{Motivation}\label{app:motivation}
Since the seminal work of \citet{gale62}, one of the central objectives of market design has been a problem of matching agents to institutions, e.g., students to schools or colleges \citep{abdulkadiroglu03}, doctors to hospitals \citep{roth99}, workers to employers \citep{kelso82}, etc. The basic agents-institutions matching model consists of a two-sided market, where on the one side we have agents each possessing strict preferences over institutions, while on the other side each institution has its quota (amount of seats to distribute among agents) together with its strict preferences over agents. In general, two-sided many-to-one matching models consider various types of real-life scenarios, including but not limited to affirmative action policies \citep{sonmez19}, soft and hard upper and lower quotas \citep{ehlers14}, and regional constraints \citep{kamada12}.

In particular, in many college admissions systems worldwide, universities do not simply allocate academic seats -- they also assign limited resources such as on-campus dormitories. These dormitories are often shared across multiple programs or institutions and subject to capacity or eligibility constraints, making them an essential part of the matching process and a key determinant of students’ choices. This naturally motivates the study of matching mechanisms that jointly allocate programs and shared resources.

A variety of real-world systems illustrate this complexity. In France, the Parcoursup platform coordinates program admissions, while dormitory spaces are limited and shared across programs and institutions.\footnote{Remark in Section 5.1 of the official description of Parcoursup 2025 states that some dormitories are shared across multiple boarding schools, requiring simultaneous coordination of several merit admission orders and dorm rankings. Although the general cutoff-adjustment mechanism is described in the annex, no explicit step-by-step procedure or numerical example is provided for the shared dormitory case. The full official up-to-date description of Parcoursup (in French) can be found at \url{https://services.dgesip.fr/T454/S764/algorithme_national_de_parcoursup}.} In Turkey, students admitted via ÖSYM apply for KYK dormitories that serve multiple universities in the same region and are allocated under a central system with social criteria.\footnote{See \url{https://studyinturkiye.com/student-housing-options-in-istanbul-a-guide-for-international-students/}.} In Russia, dormitories are allocated by universities, typically with eligibility rules and limited capacity; these dorms are frequently shared across faculties and sometimes across universities in multi-campus systems.\footnote{See \url{https://news.itmo.ru/en/features/life_in_russia/news/12721/}.}

In China, the Gaokao-based admission system is linked to dormitory assignments, which are typically program- or campus-based but constrained by capacity, often affecting students’ preferences.\footnote{See \url{https://en.sias.edu.cn/Campus_Life1/Student_Life/Dormitories.htm}.} Similarly, in India, the JoSAA mechanism allocates seats to top engineering colleges, but hostel access -- essential for many -- is institution- and category-dependent, with gender- and caste-based reservations.\footnote{See \url{https://josaa.nic.in}.}

Beyond these examples, several countries feature explicitly intercollegiate or theme-based dormitory systems. For instance, Germany offers shared student housing across universities and programs, sometimes through cooperative or merit-based selection (e.g., Studentendorf in Berlin).\footnote{See \url{https://www.studentendorf.berlin}.} In the UK, the University of London’s intercollegiate halls, such as International Hall, allocate housing across member colleges.\footnote{See \url{https://www.london.ac.uk/about/services/halls/international-hall}.} Belgium also provides a unique structure through kots-à-projet—theme-based shared flats managed jointly by students and the university.\footnote{See \url{http://www.organe.be/les-kots-a-projets/}.}

These systems share a common feature: shared, scarce dormitory resources influence both the allocation of students to programs and the structure of their preferences, especially when housing access is limited or prioritized by policy. The presence of regional, institutional, or category-specific dormitory quotas—and their interaction with academic admissions—calls for theoretical models and mechanism design tools that can account for these complementarities and constraints.

\section{Related Literature}\label{app:literature}

Groundbreaking work by \citet{kamada15} introduced matching with distributional constraints into both computer science and economics. Constraints are called \textit{aggregate} if there is at least one constraint that considers more than one college (e.g., regional constraints of \citet{kamada12} and \citet{kamada18}). RRC are obviously aggregate. The following papers work with two-sided many-to-one matching markets under various types of constraints, none of which fully capture RRC.

\citet{delacretaz23} and \citet{nguyen21} consider a two-sided many-to-one matching market under multidimensional knapsack constraints, which are not aggregate. \citet{hafalir22} study stable assignments under a model without aggregate constraints, but with regions. The authors work with choice functions of a region. \citet{liu23} present Student-Project-Resource problem with aggregate constraints. The authors introduce resources and corresponding regions, however, first, students do not have preferences over resources (a unit of a resource is necessary for a college to admit one student), and, second, it is assumed that any resource if fully indivisible, e.g., dormitory should give all its rooms to only one college. In contrast, our approach allows for student's preferences over resources, and for splitting rooms among several colleges from one region. \citet{kojima20} consider a matching problem under general constraints, where each college is constrained to choose among the predefined nonempty collection of subsets of students. \cite{romm24} show that stable allocations may admit justified envy under such constraints. \citet{kawase19} study algorithms that determine whether a given matching is stable under such constraints. \citet{imamura24b} present algorithms that check whether a given matching is Pareto-efficient under such constraints. \citet{kamada24} introduce the most general version of such constraints, named \textit{general upper bounds}. However, since the set of subsets of students to choose from is fixed for each college, such constraints are not aggregate. \citet{suzuki23} study a problem of redistribution of already admitted students across colleges, where constraints are extensionally represented by a set of school-feasible vectors (each containing a list of capacities for all colleges). Like RRC, this approach introduces dependence of one college possible chosen set from the chosen sets of other college(s), however, it has no resources. 

\citet{ismaili18} construct a model of Student-Project-Resource allocation with indivisible resources. \citet{barrot23} and \citet{yahiro24} introduce stable and strategy-proof mechanisms under union of symmetric M-convex sets constraints. However, these constraints are based on the set of school-feasible vectors, thus, do not capture RRC. \citet{kamada23} study Pareto-efficient and fair matchings under a problem of inter-regional students flow with equal in- and out- flow, but without regional constraints.